\documentclass[a4paper,11pt]{article}
\usepackage[utf8]{inputenc}
\usepackage{amssymb}
\usepackage{amsmath}
\usepackage{amsthm}
\usepackage{dsfont}
\usepackage[T1]{fontenc}
\usepackage[ngerman,english]{babel}
\usepackage[inline]{enumitem}
\usepackage{graphicx}

\numberwithin{equation}{section}

\DeclareMathOperator{\supp}{supp}
\DeclareMathOperator{\spec}{spec}
\DeclareMathOperator{\sgn}{sgn}
\DeclareMathOperator{\tr}{Tr}

\newtheorem{theorem}{Theorem}
\newtheorem{proposition}{Proposition}
\newtheorem{lemma}{Lemma}

\theoremstyle{remark}
\newtheorem{remark}{Remark}

\theoremstyle{definition}
\newtheorem{definition}{Definition}
\newtheorem{assumption}{Assumption}

\newcommand{\dda}{{\rm d}}
\newcommand{\dd}{\,\dda}
\newcommand{\ddd}[1]{\dd^{#1}}
\newcommand{\ee}{{\rm e}}

\newcommand\R{\mathbb{R}}
\newcommand\C{\mathbb{C}}
\newcommand\id{\mathds{1}}

\begin{document}

\title{Translation-invariant quasi-free states for fermionic systems\\
  and the BCS approximation\footnote{\copyright\ 2013 by the authors. This work may be reproduced, in its entirety, for non-commercial purposes.}}

\author{Gerhard Br\"aunlich$^1$, Christian Hainzl$^1$, and Robert
  Seiringer$^{2}$
  \\ \\  1. Mathematical Institute, University of T{\"u}bingen \\ 
  Auf der Morgenstelle 10, 72076 T\"ubingen, Germany \\ \\
  2. Institute of Science and Technology Austria\\ Am Campus 1, 3400 Klosterneuburg, Austria}

\date{\today}

\maketitle

\begin{abstract}
  We study translation-invariant quasi-free states for a system of
  fermions with two-particle interactions. The associated energy
  functional is similar to the BCS functional but includes also direct
  and exchange energies. We show that for suitable short-range
  interactions, these latter terms only lead to a renormalization of
  the chemical potential, with the usual properties of the BCS
  functional left unchanged. Our analysis thus represents a rigorous
  justification of part of the BCS approximation. We give bounds on
  the critical temperature below which the system displays
  superfluidity.
\end{abstract}

\section{Introduction and Main Results}
\label{sec:introduction}

The BCS theory \cite{bcs} was introduced in 1957 to describe
superconductivity, and was later extended to the context of
superfluidity 
\cite{Leggett,NRS} as a microscopic description of
fermionic gases with local pair interactions at low temperatures.  It
can be deduced from quantum physics in three steps.  One restricts the
allowed states of the system to quasi-free states, assumes
translation-invariance and $SU(2)$ rotation invariance, and finally
dismisses the direct and exchange terms in the energy.  With these
approximations, the resulting BCS functional depends, besides the
temperature $T$ and the chemical potential $\mu$, on the interaction potential $V$, the momentum distribution 
$\gamma$ and the Cooper pair wave function $\alpha$.  A non-vanishing
$\alpha$ implies a macroscopic coherence of the particles
involved, i.e., the formation of a condensate of Cooper pairs.  This motivates the
characterization of a superfluid phase by the existence of a  minimizer of
the BCS functional for which $\alpha \neq 0$.  

A rigorous treatment of the BCS functional was presented in \cite{HHSS,HS-spec_prop,HS-T_C,FHNS2007},
where the question was addressed for which interaction potentials $V$ and at which
temperatures $T$ a superfluid phase exists.  In the
present work, we focus on the question to what extent it is
justifiable to dismiss the direct and exchange terms in the energy.  A
heuristic justification was given  
in
\cite{Leggett,leggett_quantum_liquids}, where it was argued that
as long as the range of the interaction potential is suitably small,
the only effect of the direct and exchange terms is to renormalize the
chemical potential.

In this paper we derive a gap equation for the extended theory
with direct and exchange terms and investigate the existence of
non-trivial solutions for general interaction potentials.
We give a rigorous justification for dismissing the two terms for
potentials whose range $\ell$ is short compared to the scattering
length $a$ and the Fermi wave length $\frac{2\pi}{\sqrt{\mu}}$.
The potentials are required to have a suitable repulsive core to assure stability of the system. 
We show that, for small enough $\ell$, the system still
can be described by the conventional BCS equation if the chemical
potential is renormalized appropriately. In the limit $\ell\to 0$, the
spectral gap function $\Delta_\ell(p)$ converges to a constant
function  and we recover the
BCS equation in its form found in the physics literature.

While we do not prove that for fixed, finite $\ell$ there exists a
critical temperature $T_c$ such that superfluidity occurs if and only
if $T<T_c$, we find bounds $T_\ell^+$ and $T_\ell^-$ such that
$T<T_\ell^-$ implies superfluidity and $T>T_\ell^+$ excludes
superfluidity.  Moreover, in the limit $\ell \to 0$ the two bounds
converge to the same temperature, $ \lim_{\ell\to 0} T_\ell^- =
\lim_{\ell\to 0} T_\ell^+$, which can be determined by the usual BCS
gap equation.  The situation is illustrated in the following sketch.
\begin{center}\label{fig1}
  \includegraphics[width=0.25\linewidth]{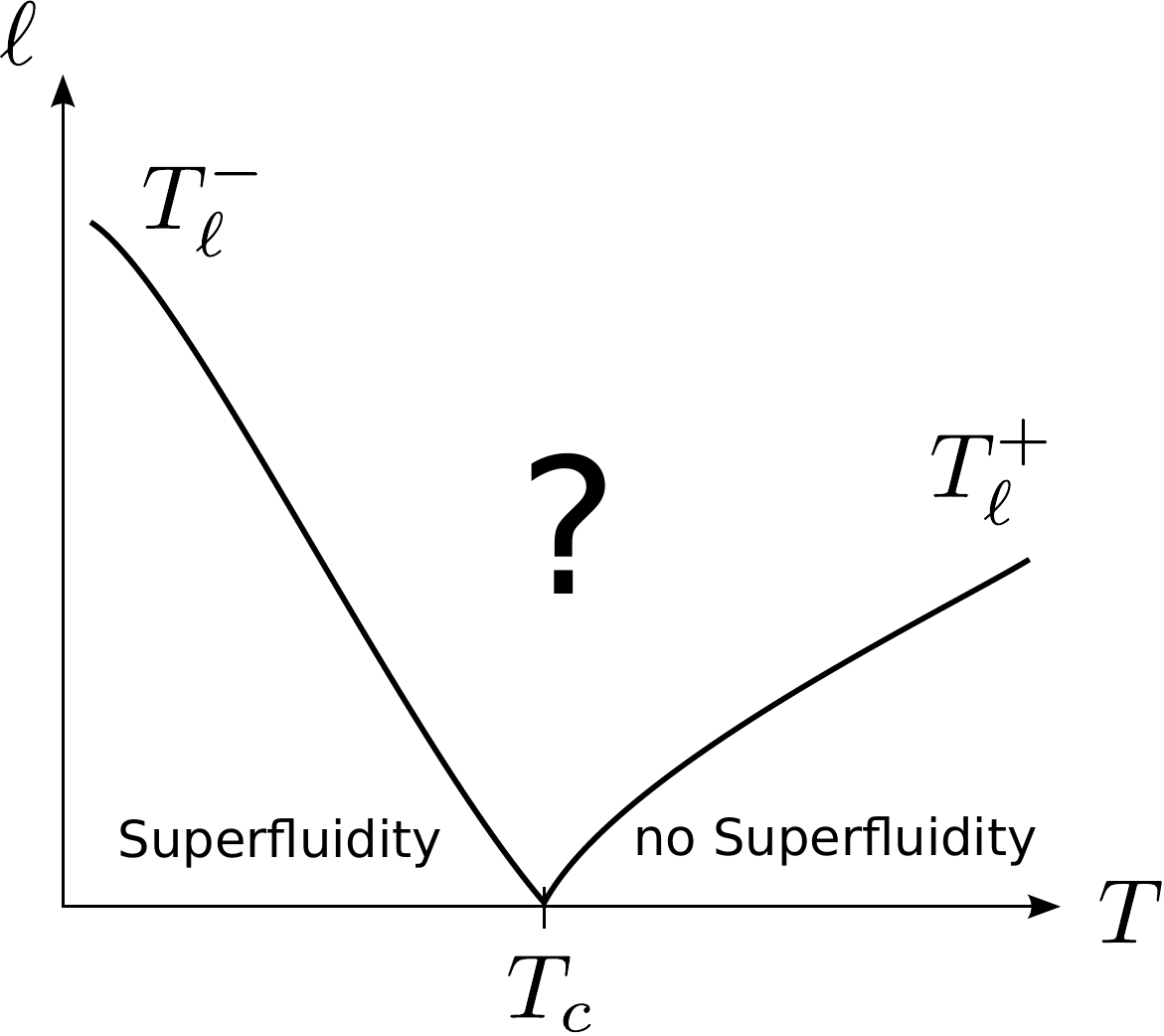}
\end{center}

We note that similar models as the one considered in our paper are
sometimes referred to as Bogoliubov-Hartree-Fock theory and have been
studied previously mainly with Newtonian interactions, modeling stars,
and without the restriction to translation-invariant states
\cite{BacFroJon-09,LenLew-12b, HaiLenLewSch-10}.  The proof of
existence of a minimizer in \cite{LenLew-12b} turns out to be
surprisingly difficult and even more strikingly, the appearance of
pairing is still open. It was confirmed numerically for the Newton
model and also for models with short range
interaction  in \cite{LewPau-12}. Hence the present work represents the
first proof of existence of pairing in a translation-invariant
Bogoliubov-Hartree-Fock model in the continuum.  For the Hubbard model at half filling
this was shown earlier in \cite{BLS}.

\subsection{The Model}
\label{sec:model}

We consider a gas of spin $1/2$ fermions in the thermodynamic limit at temperature
$T \geq 0 $ and chemical potential $\mu \in \mathbb{R}$. The particles
interact via a local two-body potential which we denote by $V$. We
assume $V$ to be reflection-symmetric, i.e., $V(-x) = V(x)$. The state
of the system is described by two functions $\hat\gamma:\R^3\to \R_+$
and $\hat\alpha:\R^3\to \C$, with $\hat\alpha(p)=\hat\alpha(-p)$, which are conveniently combined into a
$2\times 2$ matrix 
\begin{equation}\label{def:Gamma}
  \Gamma(p) = \left(
    \begin{array}{cc}
      \hat{\gamma}(p) & \hat{\alpha}(p)\\
      \overline{\hat{\alpha}(p)} & 1-\hat{\gamma}(-p)
    \end{array}
  \right),
\end{equation}
required to satisfy $0\leq \Gamma \leq \mathds{1}_{\mathbb{C}^2}$ at
every point $p\in\mathbb{R}^3$. The function $\hat\gamma$ is
interpreted as the momentum distribution of the gas, while $\alpha$
(the inverse Fourier transform of $\hat\alpha$) is the Cooper pair
wave function. Note that there are no spin variables in $\Gamma$; the
full, spin dependent Cooper pair wave function is the product of
$\alpha(x-y)$ with an antisymmetric spin singlet.

The \emph{BCS-HF functional} $\mathcal{F}_T^V$, whose infimum over all
states $\Gamma$ describes the negative of the pressure of the system,
is given as 
\begin{equation}
  \label{eq:F_T}
  \begin{split}
    \mathcal{F}_T^V(\Gamma) =& \int_{\mathbb{R}^3} (p^2 -
    \mu)\hat{\gamma}(p) \ddd{3}p +\int_{\mathbb{R}^3} |\alpha(x)|^2
    V(x) \ddd{3}x - T S(\Gamma)
    \\
    &-\int_{\mathbb{R}^3} |\gamma(x)|^2 V(x) \ddd{3}x
    +2\gamma(0)^2\int_{\mathbb{R}^3} V(x) \ddd{3}x,
  \end{split}
\end{equation}
where
\begin{equation*}
  S(\Gamma) = -\int_{\mathbb{R}^3} \tr_{\mathbb{C}^2} \big(\Gamma(p) \ln \Gamma(p)\big) \ddd{3}p
\end{equation*}
is the entropy of the state $\Gamma$. Here, $\gamma$ and $\alpha$
denote the inverse Fourier transforms of $\hat\gamma$ and
$\hat\alpha$, respectively. The last two terms in \eqref{eq:F_T} are
referred to as the \emph{exchange term} and \emph{the direct term},
respectively.  The functional \eqref{eq:F_T} can be obtained by
restricting the many-body problem on Fock space to
translation-invariant and spin-rotation invariant quasi-free states,
see \cite[Appendix A]{HHSS} and \cite{BLS}. The factor $2$ in the last
term in (\ref{eq:F_T}) originates from two possible orientations of
the particle spin.

A \emph{normal state} $\Gamma_0$ is a minimizer of the functional (\ref{eq:F_T}) 
restricted to states with $\alpha=0$. Any such minimizer can easily be
shown to be of the form
\begin{equation*}
  \hat{\gamma}_0(p) = 
  \frac{1}{1+\ee^{\frac{\varepsilon^{\gamma_0}(p) -
        \tilde{\mu}^{\gamma_0}}{T}}},
\end{equation*}
where we denote, for general $\gamma$, 
\begin{align}
  \label{eq:epsilon}
  \varepsilon^\gamma(p) &= p^2 - \frac 2{(2\pi)^{3/2}} \int_{\R^3} \left( \hat V(p-q) -\hat V(0)\right) \hat\gamma(q) \ddd{3}q \,,
\\
  \label{eq:mu}
  \tilde{\mu}^\gamma &= \mu -\frac 2{(2\pi)^{3/2}}\hat{V}(0)
  \int_{\mathbb{R}^3} \hat{\gamma}(p) \ddd{3}p\,.
\end{align}
In the absence of the exchange term the normal state would be unique, but this is not necessarily the case here. 
The system is said to be in a 
superfluid phase if and only if the minimum of $\mathcal{F}_T^V$ is
not attained at a normal state,  and we call a normal state $\Gamma_0$
\emph{unstable} in this case.

\subsection{Main Results}
\label{sec:results}

Our first goal is to characterize the existence of a superfluid phase
for a large class of interaction potentials $V$. We first find
sufficient conditions on $V$ for (\ref{eq:F_T}) to have a minimizer. These
conditions are stated in the following proposition.

\begin{proposition}[Existence of minimizers]
  \label{prop:minimizer}
Let $\mu \in \mathbb{R}$, $0 \leq T < \infty$, and let $V\in L^1(\mathbb{R}^3)\cap L^{3/2}(\mathbb{R}^3)$ be
  real-valued with $\|\hat{V}\|_\infty \leq 2
  \hat{V}(0)$.  Then
  $\mathcal{F}_T^V$ is bounded from below and attains a minimizer $(\gamma,\alpha)$ on
  \begin{equation*}
    \mathcal{D} = \left\{ \Gamma \ \text{of the form (\ref{def:Gamma})} \ \middle|
      \begin{array}{l}
        \scriptstyle{\hat{\gamma}\in L^1\displaystyle{(}\scriptstyle\mathbb{R}^3,(1+p^2)\ddd{3}p\displaystyle{)}\scriptstyle,}\\
        \scriptstyle{\alpha\in H^1(\mathbb{R}^3,\ddd{3}x),}
      \end{array}
      0\leq \Gamma \leq \id_{\C^2} 
    \right\}.
  \end{equation*}
  Moreover, the function
  \begin{equation}
    \label{eq:Delta}
    \Delta(p) =   \frac 2 {(2\pi)^{3/2}} \int_{\R^3} \hat V(p-q) \hat{\alpha}(q) \ddd{3}q
  \end{equation}
  satisfies the BCS gap equation
  \begin{equation}
    \label{eq:gap}
    \boxed{
      \frac 1 {(2\pi)^{3/2}}\int_{\mathbb{R}^3} \hat{V}(p-q)
      \frac{\Delta(q)}{K_{T,\mu}^{\gamma,\Delta}(q)} \ddd{3}q =
      -\Delta(p)
    }
  \end{equation}
\end{proposition}

In (\ref{eq:gap}) we have introduced the notation
  \begin{align}
    \label{eq:K}
    K_{T,\mu}^{\gamma,\Delta}(p) &=
    \frac{E_\mu^{\gamma,\Delta}(p)}{\tanh\big(\frac{E_\mu^{\gamma,\Delta}(p)}{2T}\big)},\\
    \label{eq:E}
    E_\mu^{\gamma,\Delta}(p) &= \sqrt{(\varepsilon^\gamma(p) -
      \tilde{\mu}^\gamma)^2 + |\Delta(p)|^2}\,,
  \end{align}
with  $\varepsilon^{\gamma}$ and $\tilde{\mu}^{\gamma}$ defined
 in \eqref{eq:epsilon} and \eqref{eq:mu}, respectively. For $T=0$, (\ref{eq:K}) is interpreted as $K_{0,\mu}^{\gamma,\Delta}(p) = E_\mu^{\gamma,\Delta}(p)$.

We note that the BCS gap equation (\ref{eq:gap}) can equivalently be written as 
  \begin{equation*}
    (K_{T,\mu}^{\gamma,\Delta} + V) \hat{\alpha} = 0,
  \end{equation*}
where $K_{T,\mu}^{\gamma,\Delta}$ is interpreted as a multiplication operator in Fourier space, and $V$ as multiplication operator in configuration space. 
This form of the equation will turn out to be useful later on.

Proposition~\ref{prop:minimizer} shows that the condition $\|\hat
V\|_\infty \leq 2 \hat V(0)$ is sufficient for stability of the
system. The simplicity of this criterion is due to the restriction to
translation-invariant quasi-free states. Without imposing
translation-invariance, the question of stability is much more subtle. Note that  $\mathcal{F}_T^V$  is not bounded from below for negative $V$, in contrast to the BCS model (where the direct and exchange terms are neglected). 

Proposition~\ref{prop:minimizer} gives no information on whether $\Delta \neq 0$. A
sufficient condition for this to happen is given in the following
theorem.

\begin{theorem}[Existence of a superfluid phase]
  \label{thm:stability}
Let $\mu \in \mathbb{R}$, $0 \leq T < \infty$, and let $V\in L^1(\mathbb{R}^3)\cap L^{3/2}(\mathbb{R}^3)$ be
  real-valued with $\|\hat{V}\|_\infty \leq 2
  \hat{V}(0)$. Let
  $\Gamma_0 = (\gamma_0,0)$ be a normal state and recall the definition of 
$K_{T,\mu}^{\gamma_0,0}(p)$ in \eqref{eq:K}--\eqref{eq:E}.
  \begin{enumerate}[label=(\roman*)]
  \item If $\inf\spec(K_{T,\mu}^{\gamma_0,0} + V) <
    0$, \label{thm:stability:1} then $\Gamma_0$ is unstable, i.e.,
    $\inf_{\Gamma\in \mathcal{D}} \mathcal{F}_T^V(\Gamma) <
    \mathcal{F}_T^V(\Gamma_0)$.
  \item If $\Gamma_0$ is unstable, then there exist  $(\gamma,\alpha) \in \mathcal{D}$, with $\alpha \neq 0$,
    such that \label{thm:stability:2} $\Delta$ defined in
    \eqref{eq:Delta} solves the BCS gap equation \eqref{eq:gap}.
  \end{enumerate}
\end{theorem}

The theorem follows from the following arguments. The operator
$K_{T,\mu}^{\gamma_0,0} + V$ naturally appears when looking at the
second derivative of $t\mapsto \mathcal{F}_T^V(\Gamma_0+t \Gamma)$ at
$t=0$. If it has a negative eigenvalues, the second derivative is
negative for suitable $\Gamma$, hence $\Gamma_0$ is unstable.  On the
other hand, an unstable normal state implies the existence of a
minimizer with $\alpha \neq 0$, which satisfies the Euler-Lagrange
equations for $\mathcal{F}_T^V$, resulting in (\ref{eq:gap}) according to Proposition~\ref{prop:minimizer}. The details are given in Section~\ref{sec:general_potentials}.

\begin{remark}\label{rem:bcs}
  In the usual BCS model, where the direct and exchange terms are neglected, the existence of a non-trivial solution 
 to
  $(K_{T,\mu}^{0,\Delta} + V) \hat{\alpha} = 0$  implies the existence of a 
  negative eigenvalue of $K_{T,\mu}^{0,0} + V$ \cite[Theorem 1]{HHSS}. This follows from the fact that 
  $K_{T,\mu}^{0,\Delta}$ is monotone in $\Delta$, i.e.,
  $K_{T,\mu}^{0,\Delta}(p) > K_{T,\mu}^{0,0}(p)$ for $\Delta \neq 0$. In particular, the system is
  superfluid if and only if the operator $K_{T,\mu}^{0,0} +V$ has a
  negative eigenvalue.  Since this operator is monotone in $T$, the
  equation
  \begin{equation*}
    \inf\spec(K_{T_c,\mu}^{0,0} +V) = 0
  \end{equation*}
  determines the critical temperature.  In the model considered here,
  where the direct and exchange terms are not neglected, the situation is more complicated. Due to the additional dependence of
  $K_{T,\mu}^{\gamma,\Delta}$ on $\gamma$, we can no longer conclude
  that $K_{T,\mu}^{\gamma,\Delta}(p) > K_{T,\mu}^{\gamma_0,0}(p)$.
  But by Theorem~\ref{thm:stability}, the solution $T$ of
  \begin{equation}\label{pdt}
    \inf\spec(K_{T,\mu}^{\gamma_0,0} +V) = 0
  \end{equation}
  still remains a lower bound for the critical temperature.
\end{remark}

\medskip

Our main result concerns the case of short-range interaction
potentials $V$, where we can recover monotonicity in $\Delta$, and
hence conclude that (\ref{pdt}) indeed defines the correct critical
temperature. More precisely, we shall consider a sequence of
potentials $\{V_\ell\}_{\ell>0}$ with $\ell \to 0$, which satisfies the following
assumptions.

\begin{assumption}  \label{asm:assumption}
\begin{enumerate}[label=(A\arabic*)]
\item \label{ax:1} $V_\ell \in L^1 \cap L^2$
\item \label{ax:supp} the range of $V_\ell$ is at most $\ell$, i.e.,  $\supp V_\ell \subseteq B_\ell(0)$
\item  \label{ax:a}
 the scattering length $a(V_\ell)$ is negative and does not
    vanish as $\ell\to 0$, i.e., $\lim_{\ell\to 0} a(V_\ell) = a < 0$
\item \label{ax:L1} $\limsup_{\ell\to 0}\|V_\ell\|_1 < \infty$
\item \label{ax:positivity} $\hat{V}_\ell(0) > 0$ and $\lim_{\ell\to 0}\hat{V}_\ell(0) = \mathcal{V} \geq 0$
\item \label{ax:3} $\|\hat{V}_\ell\|_\infty \leq 2 \hat{V}_\ell(0)$
\item \label{ax:L2} for small $\ell$, $\|V_\ell\|_2 \leq C_1 \ell^{-N}$ for some $C_1 > 0$ and $N \in \mathbb{N}$
\item \label{ax:infspec} $\exists\, 0 < b < 1$ such that $\inf\spec(p^2+V_\ell - |p|^b) > C_2 > -\infty$ holds independently of $\ell$
\item \label{ax:A} the operator $1+ V_\ell^{1/2} \frac{1}{p^2}|V_\ell|^{1/2}$ is invertible, and has an eigenvalue $e_\ell$ of order $\ell$, with corresponding eigenvector $\phi_\ell$. Moreover, $(1+
    V_\ell^{1/2}\frac{1}{p^2}|V_\ell|^{1/2})^{-1}(1-P_\ell)$ 
    is uniformly bounded in $\ell$, where
    $P_\ell = \langle J_\ell
      \phi_\ell|\phi_\ell\rangle^{-1}   |\phi_\ell\rangle
    \langle J_\ell \phi_\ell|$ and  $J_\ell = \sgn(V_\ell)$
\item \label{ax:n} the eigenvector $\phi_\ell$ satisfies
  $|\langle
    \phi_\ell| \sgn(V_\ell) \phi_\ell\rangle|^{-1} \langle |V_\ell|^{1/2}| |\phi_\ell|\rangle \leq O(\ell^{1/2})$ for small $\ell$.
\end{enumerate}
\end{assumption}

Here we use the notation $\sgn(V) =\big\{
    \begin{smallmatrix}
      1,& V \geq 0\\
      -1,& V <0
    \end{smallmatrix}$ and $V^{1/2}(x) = \sgn(V)|V(x)|^{1/2}$.
     As discussed in \cite{HS-mu}, the scattering length of a
    real-valued potential $V \in L^1(\mathbb{R}^3) \cap
    L^{3/2}(\mathbb{R}^3)$ is given by
\begin{equation}
  \label{eq:a}
  a(V) = \frac{1}{4\pi} \left\langle |V|^{1/2} \left| \tfrac{1}{1+V^{1/2}
    \frac{1}{p^2}|V|^{1/2}}\right. V^{1/2}\right\rangle\,.
\end{equation}
Assumptions \ref{ax:3}--\ref{ax:n} are to some extent technical and
are needed, among other things, to guarantee that
$\mathcal{F}_T^{V_\ell}$ is bounded from below uniformly in
$\ell$. Our main results presumably hold for a larger class of
potentials with less restrictive assumptions, but to avoid additional
complications in the proofs we do not aim here for the greatest
possible generality. Assumption~\ref{asm:assumption} implies, in particular, that $V_\ell$ converges to a point interaction as $\ell\to 0$, and we refer to \cite{BHS-delta} for a general study of point interactions arising as limits of short-range potentials of the form considered here.

\begin{remark}
  As an example for such a sequence of short-range potentials $V_\ell$ we have the following picture in mind:
  \begin{center}\label{fig2}
    \includegraphics[width=0.15\linewidth]{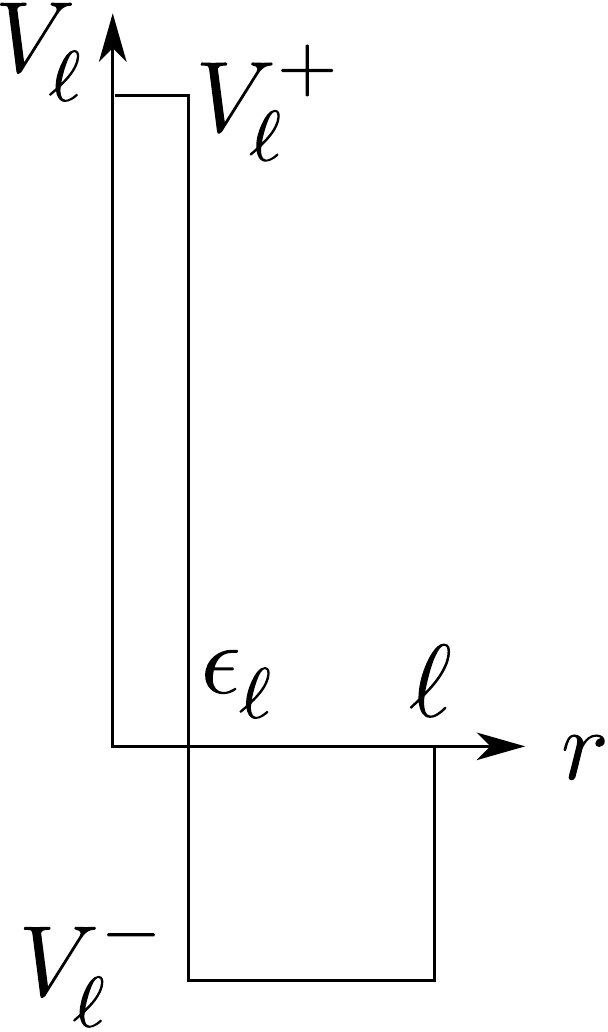}
  \end{center}
  The attractive part allows to adjust the scattering length. The
  repulsive core is needed to guarantee stability, and can be used to
  adjust the $L^1$ norm. If its range is small compared to the range
  of the attractive part, i.e., $\epsilon_\ell \ll \ell$, the
  scattering length is essentially unaffected by the repulsive core.
  In Appendix \ref{sec:example_scattering_length}, we construct an
  explicit example of such a sequence, satisfying all the assumptions 
  \ref{ax:1}--\ref{ax:n}. As $\ell\to 0$, it approximates a contact
  potential, defined via suitable selfadjoint extensions of $-\Delta$
  on $\R^3\setminus\{0\}$. Functions in its domain are known to
  diverge as $|x|^{-1}$ for small $x$, hence decay like $p^{-2}$ for
  large $|p|$.  This suggests the validity of \ref{ax:infspec} for
  $b<1$. Assumption \ref{ax:A} is easy to show in  case $V_\ell$ is uniformly bounded in $L^{3/2}$ (in which case
  $\hat{V}_\ell(0) = O(\ell)$) but much harder to prove if
  $\lim_{\ell\to 0}\hat{V}_\ell(0) > 0$.  It is possible to generalize
  \ref{ax:A} and allow finitely many eigenvalues of $1+ V_\ell^{1/2}
  \frac{1}{p^2}|V_\ell|^{1/2}$ of order $\ell$. For simplicity
  we restrict to the case of only one eigenvalue of order $\ell$, however.
\end{remark}

For the remainder of this section, we assume that the sequence $V_\ell$
satisfies \ref{ax:1}--\ref{ax:n}. We shall use the notation
$$
\tilde\mu^{\gamma_\ell} = \mu -\frac 2{(2\pi)^{3/2}}\hat{V_\ell}(0)
  \int_{\mathbb{R}^3} \hat{\gamma_\ell}(p) \ddd{3}p
$$
in analogy to (\ref{eq:mu}). 

\begin{theorem}[Effective Gap equation]
  \label{thm:gap_eff}
  Let $T \geq 0$, $\mu\in \R$, and let $(\hat{\gamma}_\ell,\hat{\alpha}_\ell)$ be a minimizer of
      $\mathcal{F}_T^{V_\ell}$ with corresponding $\Delta_\ell = 2 (2\pi)^{-3/2}\hat V_\ell * \hat\alpha_\ell$. Then there exist $\Delta \geq 0$ and $\hat\gamma:\R^3\to \R_+$ such that $|\Delta_\ell(p)| \to \Delta$ pointwise, $\hat\gamma_\ell(p) \to \hat\gamma(p)$ pointwise and $\tilde{\mu}^{\gamma_\ell} \to \tilde{\mu}$ as  
  $\ell \to 0$, satisfying 
\begin{equation}\label{eq:sat}
 \begin{aligned}
  \tilde{\mu} &= \mu
        - \frac {2\mathcal{V}}{(2\pi)^{3/2}}\int_{\mathbb{R}^3} \hat{\gamma}(p)
        \ddd{3}p \\
  \hat \gamma(p) & = \frac 12 - \frac{p^2 - \tilde\mu}{2 K_{T,\tilde\mu}^{0,\Delta}(p)} \,.
 \end{aligned}
\end{equation}
If $\Delta_\ell\neq 0$ for a subsequence of $\ell$'s going to zero, then, in addition, 
  \begin{equation}
    \label{eq:gap_eff}
    \boxed{
        -\frac{1}{4\pi a} =
        \frac{1}{(2\pi)^3}\int_{\mathbb{R}^3}\left(\frac{1}{K_{T,\tilde{\mu}}^{0,\Delta}}
          -\frac{1}{p^2}\right) \ddd{3}p\,.
      }
  \end{equation}
\end{theorem}

Recall that, according to our definitions \eqref{eq:epsilon}--\eqref{eq:E}, 
$$
 K_{T,\tilde\mu}^{0,\Delta}(p) =
    \frac{E_{\tilde\mu}^{0,\Delta}(p)}{\tanh\big(\frac{E_{\tilde\mu}^{0,\Delta}(p)}{2T}\big)}\,, \quad 
    E_{\tilde\mu}^{0,\Delta}(p) = \sqrt{(p^2 -
      \tilde{\mu})^2 + |\Delta|^2}\,.
$$

\begin{remark}
  If we consider potentials such that $\hat{V}_\ell(0) \rightarrow 0$,
  we obtain at the same time that $\tilde{\mu}^{\gamma_\ell}
  \rightarrow \mu$ and consequently (\ref{eq:gap_eff}) becomes 
  \begin{equation}
    \label{eq:gap_eff:0}
    -\frac{1}{4\pi a}
    =
    \frac{1}{(2\pi)^3}\int_{\mathbb{R}^3}\left(\frac{1}{K_{T,\mu}^{0,\Delta}}
      -\frac{1}{p^2}\right) \ddd{3}p.
  \end{equation}
Equation \eqref{eq:gap_eff:0} is the form of the BCS gap equation
  one finds in the literature, see for instance \cite{Leggett}.
\end{remark}

The effective gap equation (\ref{eq:gap_eff}) suggests to define the
critical temperature of the system via the solution of
(\ref{eq:gap_eff}) for $\Delta=0$, in which case $\hat\gamma$ is given
by $(1+\ee^{\frac{p^2-\tilde{\mu}}{T}} )^{-1}$.

\begin{definition}[Critical temperature / renormalized chemical potential]\label{def:tc}
  Let $\mu >0$.  The \emph{critical temperature} $T_c$ and the
  \emph{renormalized chemical potential} $\tilde{\mu}$ in the limit of
  a contact potential with scattering length $a<0$ and $\lim_{\ell\to 0}
  \hat{V}_\ell(0) = \mathcal{V} \geq 0$ are implicitly given by the
  set of equations
  \begin{equation}
    \label{eq:T_c}
    \begin{split}
      -\frac{1}{4\pi a} &= \frac{1}{(2\pi)^3}\int_{\mathbb{R}^3}
      \left(
        \frac{\tanh\big(\frac{p^2-\tilde{\mu}}{2T_c}\big)}{p^2-\tilde{\mu}}
        -\frac{1}{p^2} \right) \ddd{3}p
     \, ,\\
      \tilde{\mu} &= \mu -\frac{2
        \mathcal{V}}{(2\pi)^{3/2}}\int_{\mathbb{R}^3}
      \frac{1}{1+\ee^{\frac{p^2-\tilde{\mu}}{T_c}}} \ddd{3}p \,.
    \end{split}
  \end{equation}
\end{definition}

We will show existence and uniqueness of $T_c$ and $\tilde \mu$ in
Appendix~\ref{sec:appendix:tc}. Note that it is essential that
$\mu>0$. If $\mu\leq 0$, then $\tilde\mu\leq 0$ and hence the right
side of the first equation in (\ref{eq:T_c}) is always non-positive,
hence there is no solutions for $a<0$. In other words, $T_c = 0$ for
$\mu \leq 0$.

\begin{remark}
 In  \cite{HS-T_C}, the behavior of the first integral on the right side of \eqref{eq:T_c} as  $T_c\to 0$ was examined. This allows one to deduce the asymptotic behavior of $T_c$ as $a$ tends to zero, which equals
\begin{equation*}
  T_c = \tilde \mu\left( \frac{8}{\pi} \ee^{\gamma-2} +
    o(1)\right)\ee^{\frac{\pi}{2\sqrt{\tilde\mu}a}} \,,
\end{equation*}
with  $\gamma \approx 0.577$ denoting Euler's constant. Similarly, one can study the asymptotic behavior as $\mu\to 0$.  
\end{remark}

Although this definition for $T_c$ is only valid in the limit $\ell
\to 0$, it serves to make statements about upper and lower bounds on the critical temperature 
 for small (but non-zero) $\ell$, as sketched in the figure on page \pageref{fig1}.

\begin{theorem}[Bounds on critical temperature]  \label{thm:T_c}
 Let $\mu\in\R$, $T\geq 0$ and let $(\gamma_\ell^0,0)$ be a normal state for $\mathcal{F}_T^{V_\ell}$. 
\begin{enumerate}[label=(\roman*)]
\item For $T<T_c$, there exists an $\ell_0(T)>0$ such that for $\ell<\ell_0(T)$, 
$\inf\spec(K_{T,\mu}^{\gamma^0_\ell,0} + V_\ell) <
    0$. \label{thm:T_c:1} Consequently, the system is superfluid.
\item For $T>T_c$, there  exists an $\ell_0(T)>0$ such that for $\ell<\ell_0(T)$, $\mathcal{F}_T^{V_\ell}$ is minimized by a normal state. I.e., the system is not superfluid. 
\end{enumerate}
\end{theorem}

Theorem~\ref{thm:T_c} shows that Definition~\ref{def:tc} is indeed the
correct definition of the critical temperature in the limit $\ell \to
0$. In addition, it also shows that in this limit there is actually
equivalence of statements (i) and (ii) in
Theorem~\ref{thm:stability}. In particular, one recovers the linear
criterion for the existence of a superfluid phase valid in the usual
BCS model, as discussed in Remark~\ref{rem:bcs}.

\section{Proofs}
\label{sec:proofs}

\subsection{General Potentials}
\label{sec:general_potentials}

In this section we prove Proposition~\ref{prop:minimizer} and Theorem~\ref{thm:stability}. As a first step
we show that $\mathcal{F}_T^V$ is bounded from below and has a minimizer. 

\begin{lemma}
  \label{lemma:bound}
  Let $V \in L^1\cap L^{3/2}$, $\hat{V}(0) \geq 0$ and $\hat{V}(p)
  \leq 2 \hat{V}(0)$ for all $p\in\R^3$.  Then $\mathcal{F}_T^V$ is bounded from below
  and there exist a minimizer $\Gamma$ of $\mathcal{F}_T^V(\Gamma)$.
\end{lemma}
\begin{proof}
  The case without direct and exchange term was treated in
  \cite[Proposition 2]{HHSS}.  The Hartree-Fock part of the functional
  $\mathcal{F}_T^V$ gives the additional contribution
  \begin{equation*}
    -\int_{\mathbb{R}^3} |\gamma(x)|^2 V(x) \ddd{3}x
    +2\gamma(0)^2\int_{\mathbb{R}^3} V(x) \ddd{3}x
    = \frac 1 {(2\pi)^{3/2}}
    \int_{\mathbb{R}^3} \hat{\gamma}(2\hat{V}(0)-\hat{V})*\hat{\gamma} \ddd{3}p,
  \end{equation*}
  which is non-negative because of our assumption 
  $\hat{V}(p) \leq 2 \hat{V}(0)$.  Hence the same lower
  bound as in the case without direct and exchange term applies.

  To show the existence of a minimizer, it remains to check the weak
  lower semicontinuity of $\mathcal{F}_T^V$ in $L^q(\mathbb{R}^3)\times
  H^1(\mathbb{R}^3,\ddd{3}x)$ (note that
  $\hat{\gamma} \in L^1(\mathbb{R}^3)\cap L^\infty(\mathbb{R}^3)$).
  The exchange term
  \begin{equation*}
    \gamma \mapsto \int_{\mathbb{R}^3} V(x) |\gamma(x)|^2 \ddd{3} x
  \end{equation*}
  is actually weakly continuous on $H^1(\mathbb{R}^3)$, see, e.g.,
  \cite[Thm. 11.4]{LL}.  Since also $\lim_{n\to\infty}
  \int_{\mathbb{R}^3} \hat{\gamma}_n \ddd{3}p \geq \int_{\mathbb{R}^3}
  \hat{\gamma} \ddd{3}p$, the direct term is weakly lower
  semicontinuous.  In the proof of \cite[Proposition 2]{HHSS} it was
  shown that all other terms in $\mathcal{F}_T^V$ are weakly lower
  semicontinuous as well.  As a consequence, a minimizing sequence
  will actually converge to a minimizer.
\end{proof}

\begin{lemma}
  \label{lemma:EL}
  The Euler-Lagrange equations for a minimizer $(\gamma, \alpha)$ of $\mathcal{F}_T^V$ are of the form
  \begin{align}
    \label{eq:el_gamma}
    \hat{\gamma}(p) &= \frac{1}{2} - \frac{\varepsilon^\gamma(p) -
    \tilde{\mu}^{\gamma}}{2 K_{T,\mu}^{\gamma,\Delta}(p) }
    \\
    \label{eq:el_alpha}
    \hat{\alpha}(p) &= \frac{1}{2}\Delta(p)
    \frac{\tanh\big(\frac{E_\mu^{\gamma,\Delta}(p)}{2T}\big)}{E_\mu^{\gamma,\Delta}(p)},
  \end{align}
  where we used the abbreviations introduced in \eqref{eq:epsilon}--\eqref{eq:E}. In particular, the BCS gap equation (\ref{eq:gap}) holds. 
\end{lemma}

\begin{proof}
  The proof works similar to \cite{HHSS}. We sketch here an
  alternative, more concise derivation, restricting our attention to
  $T>0$ for simplicity.  A minimizer $\Gamma = (\gamma,\alpha)$
  of $\mathcal{F}_T^V$ fulfills the inequality
  \begin{equation}
    \label{eq:minimizer}
    0 \leq \left. \frac{\dda}{\dda t} \right|_{t=0} \mathcal{F}_T^V\big(\Gamma +
    t(\tilde{\Gamma}-\Gamma)\big)
  \end{equation}
  for arbitrary $\tilde{\Gamma} \in \mathcal{D}$.
  Here we may assume that $\Gamma$ stays away from $0$ and $1$ by
  arguing as in \cite[Proof of Lemma 1]{HHSS}.
  A simple
  calculation using
  \begin{equation*}
    S(\Gamma) =
    -\int_{\mathbb{R}^3}\tr_{\mathbb{C}^2}\Gamma\ln\Gamma \ddd{3}p
    =-\frac{1}{2}\int_{\mathbb{R}^3}\tr_{\mathbb{C}^2}\big(\Gamma\ln(\Gamma)+(1-\Gamma)\ln(1-\Gamma)\big)
    \ddd{3}p
  \end{equation*}
  shows that
  \begin{align*}
    \left. \frac{\dda}{\dda t} \right|_{t=0}
    \mathcal{F}_T^V\big(\Gamma + t(\tilde{\Gamma}-\Gamma)\big) =
    \frac{1}{2}\int_{\mathbb{R}^3}\tr_{\mathbb{C}^2}\left[(H_\Delta
      (\tilde{\Gamma}-\Gamma)\big)\ddd{3}p + T (\tilde{\Gamma}-\Gamma)
      \ln\Big(\frac{\Gamma}{1-\Gamma}\Big) \right]\ddd{3}p,
  \end{align*}
  with
  \begin{equation*}
    H_\Delta = \left(
      \begin{array}{ll}
        \varepsilon^{\gamma} - \tilde{\mu}^{\gamma} & \Delta \\
        \bar{\Delta} & -(\varepsilon^{\gamma} - \tilde{\mu}^{\gamma})
      \end{array}
    \right),
  \end{equation*}
  using the definition
  \begin{equation*}
    \Delta = 2 (2\pi)^{-3/2} \hat{V} * \hat{\alpha}.
  \end{equation*}
  Separating the terms containing no $\tilde{\Gamma}$ and moving them
  to the left side in \eqref{eq:minimizer}, we obtain
  \begin{equation*}
    \int_{\mathbb{R}^3} \tr_{\mathbb{C}^2}\Big(H_\Delta \big(\Gamma-
    \left(
      \begin{smallmatrix}
        0 & 0\\ 0& 1
      \end{smallmatrix}\right)
    \big) + T\ \Gamma
    \ln\Big(\frac{\Gamma}{1-\Gamma}\Big)\Big)\ddd{3}p
    \leq
    \int_{\mathbb{R}^3} \tr_{\mathbb{C}^2}\Big(H_\Delta \big(\tilde{\Gamma}-
    \left(
      \begin{smallmatrix}
        0 & 0\\ 0& 1
      \end{smallmatrix}\right)
    \big) + T\ \tilde{\Gamma}
    \ln\Big(\frac{\Gamma}{1-\Gamma}\Big)\Big)\ddd{3}p.
  \end{equation*}
  Note that $\int_{\mathbb{R}^3}\tr_{\mathbb{C}^2}\big(H_\Delta
  \Gamma\big)\ddd{3}p$ is not finite but $ \int_{\mathbb{R}^3}
  \tr_{\mathbb{C}^2}\Big(H_\Delta \big(\tilde{\Gamma}- \left(
    \begin{smallmatrix}
      0 & 0\\ 0& 1
    \end{smallmatrix}\right)
  \big)\Big)\ddd{3}p$ is.  Since $\tilde{\Gamma}$ was arbitrary,
  $\Gamma$ also minimizes the linear functional
  \begin{equation*}
    \tilde{\Gamma} \mapsto
    \int_{\mathbb{R}^3} \tr_{\mathbb{C}^2}\Big(H_\Delta \big(\tilde{\Gamma}-
    \left(
      \begin{smallmatrix}
        0 & 0\\ 0& 1
      \end{smallmatrix}\right)
    \big) + T\ \tilde{\Gamma}
    \ln\Big(\frac{\Gamma}{1-\Gamma}\Big)\Big)\ddd{3}p,
  \end{equation*}
  whose Euler-Lagrange equation is of the simple form
  \begin{equation}
    \label{eq:el_Gamma}
    0 = H_\Delta + T \ln\left(\frac{\Gamma}{1-\Gamma}\right),
  \end{equation}
  which is equivalent to
  \begin{equation*}
    \Gamma = \frac{1}{1+\ee^{\frac{1}{T}H_\Delta}}.
  \end{equation*}
  This in turn implies \eqref{eq:el_gamma} and \eqref{eq:el_alpha}.
  Indeed, $H_\Delta^2 = [E_\mu^{\gamma,\Delta}]^2\,
  \mathds{1}_{\mathbb{C}^2}$ and, therefore,
  \begin{equation*}
    \ee^{\frac{1}{T}H_\Delta} =
    \cosh\Big(\frac{1}{T}E_\mu^{\gamma,\Delta}\Big)\mathds{1}_{\mathbb{C}^2} +
    \frac{1}{E_\mu^{\gamma,\Delta}}
    \sinh\Big(\frac{1}{T}E_\mu^{\gamma,\Delta}\Big) H_\Delta.
  \end{equation*}
  With the relations
  \begin{align*}
    1+\cosh(x) &= \frac{\sinh(x)}{\tanh(x/2)}, \\
    1-\cosh(x) &= -\tanh(x/2)\sinh(x),
  \end{align*}
  we see that
  \begin{equation*}
    (1+\ee^{\frac{1}{T}H_\Delta})H_\Delta
    =
    -K_{T,\mu}^{\gamma,\Delta}(1-\ee^{\frac{1}{T}H_\Delta}),\qquad
    \textrm{ where }\quad
    K_{T,\mu}^{\gamma,\Delta} =
    \frac{E_\mu^{\gamma,\Delta}}{\tanh(\frac{1}{2T}E_\mu^{\gamma,\Delta})}.
  \end{equation*}
  Consequently,
  \begin{equation*}
    \Gamma = \frac{1}{1+\ee^{\frac{1}{T}H_\Delta}}
    = \frac{1}{2} +
    \frac{1}{2}\frac{1-\ee^{\frac{1}{T}H_\Delta}}{1+\ee^{\frac{1}{T}H_\Delta}}
    = \frac{1}{2} -\frac{1}{2 K_{T,\mu}^{\gamma,\Delta}}H_\Delta
    = \left(
      \begin{array}{ll}
        \frac{1}{2} - \frac{\varepsilon^{\gamma}-\tilde{\mu}^{\gamma}}{2 K_{T,\mu}^{\gamma,\Delta}} &
        -\frac{\Delta}{2 K_{T,\mu}^{\gamma,\Delta}} \\
        -\frac{\bar{\Delta}}{2 K_{T,\mu}^{\gamma,\Delta}} & \frac{1}{2} +
        \frac{\varepsilon^{\gamma}-\tilde{\mu}^{\gamma}}{2 K_{T,\mu}^{\gamma,\Delta}}
      \end{array}
    \right).
  \end{equation*}
\end{proof}

\begin{proof}[Proof of Proposition~\ref{prop:minimizer}]
  This is an immediate consequence of Lemmas~\ref{lemma:bound} and~\ref{lemma:EL}.
\end{proof}

\begin{proof}[Proof of Theorem~\ref{thm:stability}]
  The proof works exactly as the analog steps in \cite[Proof of
  Theorem 1]{HHSS}.  To see \ref{thm:stability:1}, note that
  $\langle \hat{\alpha}, (K_{T,\mu}^{\gamma_0,0} +
  V)\hat{\alpha}\rangle$ is the second derivative of $\mathcal{F}_T^V$
  with respect to $\alpha$ at $\Gamma=\Gamma_0$. For
  \ref{thm:stability:2}, we use the fact that the gap equation is a
  combination of the Euler-Lagrange equation \eqref{eq:el_alpha} of
  the functional and the definition of $\Delta$.
\end{proof}

\subsection{Sequence of Short-Range Potentials}
\label{sec:contactpotential}

In the following we consider a sequence of potentials $V_\ell$
satisfying the assumptions
\ref{ax:1}--\ref{ax:n} in Assumption~\ref{asm:assumption}. Since $V_\ell$ converges to a
contact potential, Lemma~\ref{lemma:bound} is not sufficient to prove
that $\mathcal{F}_T^{V_\ell}$ is uniformly bounded from below. To this aim,
we have to use more subtle estimates involving bounds on the relative
entropy obtained in \cite{FHSS-micro_ginzburg_landau}, and we heavily
rely on assumption \ref{ax:infspec}.

\begin{lemma}
  \label{lemma:minimizer}
  There exists $C_1 > 0$, independent of $\ell$,  such that 
  \begin{equation}
    \label{eq:F_T_bound}
    \mathcal{F}_T^{V_\ell}(\Gamma)
    \geq -C_1
    + \frac{1}{2}\int_{\mathbb{R}^3} (1+p^2)(\hat{\gamma} - \hat{\gamma}_0)^2 \ddd{3}p
    + \frac{1}{2}\int_{\mathbb{R}^3} |p|^b |\hat{\alpha}|^2 \ddd{3}p,
  \end{equation}
  where we denote $\hat{\gamma}_0(p) =
  \frac{1}{1+\ee^{(p^2-\mu)/T}}$.
\end{lemma}

\begin{proof}
  We rewrite $\mathcal{F}_T^{V_\ell}(\Gamma)$ as
  \begin{align*}
    \mathcal{F}_T^{V_\ell}(\Gamma) = &\frac{1}{2} \int_{\mathbb{R}^3}
    \tr_{\mathbb{C}^2} \Big(H_0 \big(\Gamma- \left(
      \begin{smallmatrix}
        0 & 0\\ 0& 1
      \end{smallmatrix}\right)
    \big)\Big)\ddd{3}p
    + \int_{\mathbb{R}^3} V_\ell(x) |\alpha(x)|^2 \ddd{3}x\\
    &- T\, S(\Gamma) +  \frac 1{(2\pi)^{3/2}} \int_{\mathbb{R}^3}
    \big((2\hat{V_\ell}(0)-\hat{V_\ell})
    *\hat{\gamma}\big)(p)\,\hat{\gamma}(p) \ddd{3}p\,,
  \end{align*}
  where $\Gamma = \Gamma(\gamma, \alpha)$ and
  \begin{align*}
    H_0 &= \left(
      \begin{array}{ll}
        p^2-\mu & 0 \\
        0 & -(p^2-\mu)
      \end{array}
    \right).
  \end{align*}

  Since $\hat{\gamma}(p) \geq 0$ and, by assumption
  \ref{ax:3}, 
  $2\hat{V_\ell}(0)-\hat{V_\ell}(p)  \geq 0$, the
  combination of direct plus exchange term is non-negative and it
  suffices to find a lower bound for
  \begin{align*}
    \tilde{\mathcal{F}}_T^{V_\ell}(\Gamma) = \frac{1}{2}
    \int_{\mathbb{R}^3} \tr_{\mathbb{C}^2} \Big(H_0 \big(\Gamma-
    \left(
      \begin{smallmatrix}
        0 & 0\\ 0& 1
      \end{smallmatrix}\right)
    \big)\Big)\ddd{3}p + \int_{\mathbb{R}^3} V_\ell(x) |\alpha(x)|^2
    \ddd{3}x - T\, S(\Gamma).
  \end{align*}

  We compare $\tilde{\mathcal{F}}_T^{V_\ell}(\Gamma)$ to the value
  $\tilde{\mathcal{F}}_T^{V_\ell}(\Gamma_0)$, where $\Gamma_0 =
  \frac{1}{1+\ee^{H_0/T}}$. Their difference equals 
  \begin{equation*}
    \begin{split}
      \tilde{\mathcal{F}}_T^{V_\ell}(\Gamma) -
      \tilde{\mathcal{F}}_T^{V_\ell}(\Gamma_0) = &\frac{1}{2}
      \int_{\mathbb{R}^3} \tr_{\mathbb{C}^2} \Big(H_0
      (\Gamma-\Gamma_0)\Big)\ddd{3}p
      - T\Big(S(\Gamma)-S(\Gamma_0)\Big)\\
      &+ \int_{\mathbb{R}^3} V_\ell(x) |\alpha(x)|^2 \ddd{3}x.
    \end{split}
  \end{equation*}
  Using $H_0 + T\ln\big(\frac{\Gamma_0}{1-\Gamma_0}\big) = 0$ in the
  trace and performing some simple algebraic transformations, we may
  write
  \begin{equation*}
    \tilde{\mathcal{F}}_T^{V_\ell}(\Gamma) -
    \tilde{\mathcal{F}}_T^{V_\ell}(\Gamma_0)
    = \frac{T}{2}\mathcal{H}(\Gamma,\Gamma_0) + \int_{\mathbb{R}^3} V_\ell(x) |\alpha(x)|^2 \ddd{3}x,
  \end{equation*}
  where
  \begin{equation*}
    \mathcal{H}(\Gamma,\Gamma_0) =
    \int_{\mathbb{R}^3}\tr_{\mathbb{C}^2}\Big[\Gamma\big(\ln(\Gamma)-\ln(\Gamma_0)\big)
    + \big(1-\Gamma\big)\big(\ln(1-\Gamma) - \ln(1-\Gamma_0)\big)\Big]\ddd{3}p
  \end{equation*}
  denotes the relative entropy of $\Gamma$ and $\Gamma_0$.  Lemma~3 in
  \cite{FHSS-micro_ginzburg_landau}, which is an extension of
  Theorem~1 in \cite{HLS2008}, implies the lower bound
  \begin{align*}
    \frac{T}{2}\mathcal{H}(\Gamma,\Gamma_0) &\geq
    \frac{1}{2}\int_{\mathbb{R}^3}\tr_{\mathbb{C}^2}
    \left[\frac{H_0}{\tanh(\frac{H_0}{2T})}(\Gamma-\Gamma_0)^2
    \right]\ddd{3}p\\
    &= \int_{\mathbb{R}^3} K_{T,\mu}^{0,0}(p)
    \big((\hat{\gamma}(p)-\hat{\gamma}_0(p))^2+|\hat{\alpha}(p)|^2\big)\ddd{3}p.
  \end{align*}
  Hence we obtain
  \begin{equation*}
    \tilde{\mathcal{F}}_T^{V_\ell}(\Gamma) - \tilde{\mathcal{F}}_T^{V_\ell}(\Gamma_0)
    \geq
    \left\langle {\alpha} \left| K_{T,\mu}^{0,0}+V_\ell \right| {\alpha}
    \right\rangle + \int_{\mathbb{R}^3} K_{T,\mu}^{0,0}(p)
    (\hat{\gamma}(p)-\hat{\gamma}_0(p))^2\ddd{3}p.
  \end{equation*}
  In both terms, we can use $K_{T,\mu}^{0,0} \geq p^2-\mu$, therefore
  \begin{align*}
    \tilde{\mathcal{F}}_T^{V_\ell}(\Gamma) -
    \tilde{\mathcal{F}}_T^{V_\ell}(\Gamma_0) &\geq \left\langle
    {\alpha} \left| p^2+V_\ell-\mu \right| {\alpha} \right\rangle \\
    &\quad +  \frac{1}{2}\int_{\mathbb{R}^3} (1+p^2)(\hat{\gamma} -
    \hat{\gamma}_0)^2 \ddd{3}p + \int_{\mathbb{R}^3}
    \Big(\frac{p^2}{2}-\mu-\frac{1}{2}\Big)(\hat{\gamma} -
    \hat{\gamma}_0)^2 \ddd{3}p\,.
  \end{align*}
  Using $(\hat{\gamma} - \hat{\gamma}_0)^2 \leq 1$, we can bound
  \begin{align*}
    \int_{\mathbb{R}^3}
    \Big(\frac{p^2}{2}-\mu-\frac{1}{2}\Big)(\hat{\gamma} -
    \hat{\gamma}_0)^2 \ddd{3}p \geq
    -\int_{\mathbb{R}^3}\left[\frac{p^2}{2}-\mu-\frac{1}{2}\right]_-
    \ddd{3}p\,,
  \end{align*}
where $[t]_- = \max\{0,-t\}$ denote the negative part of a real number $t$.
  By assumption \ref{ax:infspec}, $\inf\spec
  (p^2+V_\ell-|p|^b)$ is bounded by some number $C$ independent of $\ell$. Thus
  \begin{align*}
    \int_{\mathbb{R}^3}
    \big(p^2+V_\ell-\mu\big)|\hat{\alpha}|^2\ddd{3}p &\geq
    \int_{\mathbb{R}^3}
    (|p|^b+C-\mu)|\hat{\alpha}|^2\ddd{3}p\\
    &= \frac{1}{2}\int_{\mathbb{R}^3}|p|^b|\hat{\alpha}|^2\ddd{3}p +
    \int_{\mathbb{R}^3}\Big(\frac{|p|^b}{2}+C-\mu\Big)|\hat{\alpha}|^2\ddd{3}p.
  \end{align*}
  With $|\hat{\alpha}|^2 \leq 1$ we conclude
  \begin{equation*}
    \int_{\mathbb{R}^3}\Big(\frac{|p|^b}{2}+C-\mu\Big)|\hat{\alpha}|^2\ddd{3}p
    \geq -\int_{\mathbb{R}^3}\left[\frac{|p|^b}{2}-\mu+C\right]_- \ddd{3}p.
  \end{equation*}
Our final lower bound is thus
  \begin{equation*}
    \mathcal{F}_T^{V_\ell}(\Gamma)
    \geq \tilde{\mathcal{F}}_T^{V_\ell}(\Gamma)
    \geq -C_1
    + \frac{1}{2}\int_{\mathbb{R}^3} (1+p^2)(\hat{\gamma} - \hat{\gamma}_0)^2 \ddd{3}p
    + \frac{1}{2}\int_{\mathbb{R}^3} |p|^b |\hat{\alpha}|^2 \ddd{3}p,
  \end{equation*}
  with
  \begin{align*}
    C_1 &= -\tilde{\mathcal{F}}_T^{V_\ell}(\Gamma_0) +
    \int_{\mathbb{R}^3}\left[\frac{p^2}{2}-\mu-\frac{1}{2}\right]_-
    \ddd{3}p + \int_{\mathbb{R}^3}\left[\frac{|p|^b}{2}-\mu+C\right]_-
    \ddd{3}p\,.
  \end{align*}
  Since $\tilde{\mathcal{F}}_T^{V_\ell}(\Gamma_0)$ does not depend on 
  $\ell$ (the off-diagonal entries of $\Gamma_0$ being $0$) this concludes the proof.
\end{proof}

\begin{lemma}
  \label{lemma:gamma_bound}
  If $(\gamma_\ell,\alpha_\ell)$ is a minimizer of
  $\mathcal{F}_T^{V_\ell}$, then $\int_{\mathbb{R}^3}
  \hat{\gamma}_\ell(p) |p|^b \ddd{3}p$ is uniformly bounded in $\ell$.
\end{lemma}
\begin{proof}
  To simplify notation, we leave out the index $\ell$.  A minimizer
  $(\gamma,\alpha)$ of $\mathcal{F}_T^{V}$ satisfies the
  Euler-Lagrange equation \eqref{eq:el_gamma}. Using the abbreviation
  \begin{equation*}
    K_{T,\mu}^{\gamma,\Delta} =
    \frac{E_\mu^{\gamma,\Delta}(p)}{\tanh\big(\frac{E_\mu^{\gamma,\Delta}(p)}{2T}\big)},
  \end{equation*}
  we may express \eqref{eq:el_gamma} in the form
  \begin{equation*}
    \hat{\gamma} = \frac{1}{2} - \frac{1}{2}
    \frac{\varepsilon^{\gamma}-\tilde{\mu}^\gamma}{K_{T,\mu}^{\gamma,\Delta}}.
  \end{equation*}
  Adding and subtracting
  $\frac{1}{2}\frac{E_\mu^{\gamma,\Delta}}{K_{T,\mu}^{\gamma,\Delta}}
  = \frac{1}{2}\tanh\big(\frac{E_\mu^{\gamma,\Delta}(p)}{2T}\big)$, we
  may write
  \begin{align}
    \label{eq:el_gamma:2}
    \hat{\gamma} &= \frac{1}{2}\bigg(1 -
    \tanh\Big(\frac{E_\mu^{\gamma,\Delta}}{2T}\Big)\bigg) +
    \frac{1}{2}
    \frac{E_\mu^{\gamma,\Delta}-(\varepsilon^{\gamma}-\tilde{\mu}^\gamma)}{K_{T,\mu}^{\gamma,\Delta}}\\
    \nonumber &= \frac{1}{1 + \ee^{\frac{E_\mu^{\gamma,\Delta}}{T}}}
    +\frac{1}{2}
    \frac{|\Delta|^2}{\big(E_\mu^{\gamma,\Delta}+(\varepsilon^{\gamma}-\tilde{\mu}^\gamma)\big)K_{T,\mu}^{\gamma,\Delta}}.
  \end{align}
  Using the Euler-Lagrange equation $\Delta =
  2K_{T,\mu}^{\gamma,\Delta} \hat{\alpha}$ for $\alpha$, we obtain
  \begin{equation}
    \label{eq:el_gamma:3}
    \begin{split}
      \hat{\gamma} = \frac{1}{1 +
        \ee^{\frac{E_\mu^{\gamma,\Delta}}{T}}} +2
      \frac{|\hat{\alpha}|^2K_{T,\mu}^{\gamma,\Delta}}{E_\mu^{\gamma,\Delta}+(\varepsilon^{\gamma}-\tilde{\mu}^\gamma)}.
    \end{split}
  \end{equation}
  Assumption \ref{ax:3} implies that $\varepsilon^\gamma -
  \tilde{\mu}^\gamma \geq p^2-\mu$.  In particular, the contribution of the first term is bounded by 
  \begin{equation*}
    \int_{\mathbb{R}^3} \frac{1}{1 +
      \ee^{\frac{E_\mu^{\gamma,\Delta}}{T}}} |p|^b \ddd{3}p
    \leq
    \int_{\mathbb{R}^3} \frac{1}{1 +
      \ee^{\frac{p^2-\mu}{T}}} |p|^b \ddd{3}p
  \end{equation*}
  which is independent of $\ell$.  To treat the second term, we split
  the domain of integration $\mathbb{R}^3$ into two disjoint sets and
  show that the integral is uniformly bounded on each
  subset.  On the set $B= \{p|\tanh(\frac{p^2-\mu}{2T}) \geq
  \frac{2}{3}\}$ we have that $\tanh(\frac{\varepsilon^\gamma -
    \tilde{\mu}^\gamma}{2T}) \geq \frac{2}{3}$ and $\varepsilon^\gamma
  - \tilde{\mu}^\gamma \geq 0$. This implies that
  \begin{equation*}
    \frac{|\hat{\alpha}|^2
      K_{T,\mu}^{\gamma,\Delta}}{E_\mu^{\gamma,\Delta}+(\varepsilon^{\gamma}-\tilde{\mu}^\gamma)}
    \leq
    \frac{3}{2}|\hat{\alpha}|^2,
  \end{equation*}
  whose integral over $B$ is bounded uniformly in $\ell$ by
  \eqref{eq:F_T_bound}, even after multiplication by $|p|^b$.  The
  complement $B^c=\{p|\tanh(\frac{p^2-\mu}{2T}) \leq \frac{2}{3}\}$ of
  $B$ is compact and thus also
  $\int_{B^c}\hat{\gamma}(p)|p|^b\ddd{3}p$ is trivially bounded,
  because $0\leq \hat{\gamma} \leq 1$.
\end{proof}

In the following lemma we show that, as $\ell\to 0$, pointwise limits
for the main quantities exist.  In the case of $\Delta_\ell$, observe that
$\check{\Delta}_\ell(x) =2 V_\ell(x) \alpha_\ell(x)$ is supported in
$|x| \leq \ell$. Heuristically, if the norm
$\|\check{\Delta}_\ell\|_1$ stays finite, $\check{\Delta}_\ell$ should
converge to a $\delta$ distribution and its Fourier transform
$\Delta_\ell$ to a constant function. While we do not show that
$\|\check{\Delta}_\ell\|_1$ stays finite, we can use assumption
\ref{ax:L2} to at least show that it cannot increase too fast as $\ell
\to 0$, which will turn out to be sufficient.  The pointwise
convergence $\gamma_\ell(p) \rightarrow \gamma(p)$ then follows from
Lemma~\ref{lemma:gamma_bound} together with the Euler-Lagrange
equation \eqref{eq:el_gamma} for $\gamma_\ell$.

In the following, we use the definition
\begin{equation*}
  m_\mu^{\gamma_\ell,\Delta_\ell}(T) =
  \frac{1}{(2\pi)^3}\int_{\mathbb{R}^3}\left(\frac{1}{K_{T,\mu}^{\gamma_\ell,\Delta_\ell}}
  -\frac{1}{p^2}\right) \ddd{3}p\,.
\end{equation*}

\begin{lemma}
  \label{lemma:convergence}
  Let $(\gamma_\ell,\alpha_\ell)$ be a sequence of minimizers of
  $\mathcal{F}_T^{V_\ell}$ and $\Delta_\ell = 2 (2\pi)^{-3/2} \hat{V}_\ell *
  \hat{\alpha}_\ell$. Then there are subsequences of $\gamma_\ell$ and
  $\alpha_\ell$, which we continue to denote by $\gamma_\ell$ and 
  $\alpha_\ell$, and $\gamma \in L^1(\mathbb{R}^3)\cap L^\infty(\R^3)$,
  $\Delta \in \mathbb{R}_+$ such that
  \begin{enumerate}[label=(\roman*)]
  \item $|\Delta_\ell(p)|$ converges pointwise to the constant
    function $\Delta$ as $\ell\to 0$,

  \item $\displaystyle \lim_{\ell\to 0}\int_{\mathbb{R}^3}
    \hat{\gamma}_\ell \ddd{3}p = \int_{\mathbb{R}^3} \hat{\gamma}
    \ddd{3}p$,
  \item $\displaystyle \lim_{\ell\to 0} \tilde{\mu}^{\gamma_\ell} =
    \tilde{\mu}^{\gamma}$, where $\tilde{\mu}^\gamma = \mu
    -2(2\pi)^{-3/2}\mathcal{V} \int_{\mathbb{R}^3} \hat{\gamma}(p)
    \ddd{3}p$,
  \item $\varepsilon^{\gamma_\ell}(p) \to p^2$ pointwise as $\ell\to
    0$,
 \item $\hat{\gamma}_\ell(p) \rightarrow \hat{\gamma}(p)$ pointwise as $\ell\to
    0$, and Eq.~\eqref{eq:sat} 
is satisfied for  $(\gamma, \tilde\mu^\gamma, \Delta)$,
  \item $\displaystyle\lim_{\ell\to 0}
    m^{\gamma_\ell,\Delta_\ell}_\mu(T) = m^{\gamma,\Delta}_\mu(T) =
    m^{0,\Delta}_{\tilde{\mu}^\gamma}(T)$. \label{lemma:convergence:m}
  \end{enumerate}
\end{lemma}

We shall see later that it is not necessary to restrict to a subsequence, the result holds in fact for the whole sequence. 

\begin{proof}
  (i) \label{lemma:convergence:Delta} 
    Lemma~\ref{lemma:minimizer} and Assumption \ref{ax:L2} imply that, 
    with $\check{\Delta}_\ell = 2 V_\ell \alpha_\ell$,
    \begin{equation}
      \label{eq:Delta_bound}
      \|\check{\Delta}_\ell\|_1 \leq 2 \|V_\ell\|_2\|\alpha_\ell\|_2 \leq C \ell^{-N}\,.
    \end{equation}
    The fact that $\check{\Delta}_\ell$ is compactly supported in
    $B_\ell(0)$ will allow us to argue that a suitable subsequence of
    $\Delta_\ell(p)$ converges to a polynomial in $p$. Furthermore,
    the fact that $\hat{\alpha}_\ell = -2
    (K_{T,\mu}^{\gamma_\ell,\Delta_\ell})^{-1} \Delta_\ell$ is
    uniformly bounded in $L^2$ forces the polynomial to be a constant.

    We denote by 
    \begin{equation*}
      P_{\ell,N}(p) = \frac{1}{(2\pi)^{3/2}}
      \sum_{j=0}^N
      \frac{(-i)^j}{j!} \sum_{i_1,\dotsc,i_j =1}^3
      c^{(\ell,j)}_{i_1,\dotsc,i_j} p_{i_1}\cdots p_{i_j}
    \end{equation*}
    the $N$-th order Taylor approximation of $\Delta_\ell(p) = (2\pi)^{-3/2} 
    \int_{\mathbb{R}^3} \check{\Delta}_\ell(x) \ee^{-i p \cdot x
    }\ddd{3}x$ at $p=0$, with coefficients 
    given by
    \begin{equation*}
      c^{(\ell,j)}_{i_1,\dotsc,i_j} = \int_{\mathbb{R}^3} \check{\Delta}_\ell(x)
      x_{i_1}\cdots x_{i_j} \ddd{3}x.
    \end{equation*}
    Using that $\check{\Delta}_\ell$ is supported in $B_\ell(0)$ we may
    estimate the remainder term as 
    \begin{equation*}
      | \Delta_\ell(p) - P_{\ell,N}(p)|
      =\frac{1}{(2\pi)^{3/2}}\left|\int_{\mathbb{R}^3}\check{\Delta}_\ell(x)\left(e^{-ip\cdot
        x}-\sum_{j=0}^N \frac{(-ip\cdot x)^j}{j!}\right) \ddd{3}x \right|
      \leq \frac{\ell^{N+1}}{(2\pi)^{3/2}} 
      \|\check{\Delta}_\ell\|_1  |p|^{N+1} \ee^{\ell |p|},
    \end{equation*}
    which goes to zero pointwise for $\ell\to 0$ by
    \eqref{eq:Delta_bound}.

    Now let $\bar{c}_\ell = \max_{0\leq j \leq N}\max_{1\leq
      i_1,\dotsc,i_j \leq 3} \{ |c^{(\ell,j)}_{i_1,\dotsc,i_j}| \}$.
    We want to show 
that $\bar{c} =
    \limsup_{\ell\to 0}\bar{c}_\ell < \infty$.  If $\bar{c} = 0$, we
    are done. If not, then there is a subsequence of
    $P_{\ell,N}(p)/\bar{c}_\ell$ which converges pointwise to some
    polynomial $P(p)$ of degree $n \leq N$. 
    We now use the uniform boundedness of $2\|\alpha_\ell\|_2 =
    \left\|\frac{\Delta_\ell}{K_{T,\mu}^{\gamma_\ell,\Delta_\ell}}\right\|_2$
    to conclude that $P(p)$ cannot be a polynomial of degree $n\geq 1$, 
    and that $\bar{c}$ is finite.  We first rewrite
    $2\hat{\alpha}_\ell$ as
    \begin{equation*}
      \frac{\Delta_\ell}{K_{T,\mu}^{\gamma_\ell,\Delta_\ell}}
      =
      \frac{\Delta_\ell}{E_{\mu}^{\gamma_\ell,\Delta_\ell}}\tanh\big(E_\mu^{\gamma_\ell,\Delta_\ell}/(2T)\big)
      = \frac{\Delta_\ell}{E_{\mu}^{\gamma_\ell,\Delta_\ell}}-
      \frac{\Delta_\ell}{E_\mu^{\gamma_\ell,\Delta_\ell}}
      \frac{2}{1+ \exp(E_\mu^{\gamma_\ell,\Delta_\ell}/T)}.
    \end{equation*}
    Using $E_\mu^{\gamma_\ell,\Delta_\ell} \geq
    \varepsilon^{\gamma}-\tilde{\mu}^\gamma \geq p^2-\mu$ and
    $|\Delta_\ell| \leq E_\mu^{\gamma_\ell,\Delta_\ell}$,
    it is easy to see that the
    $L^2$ norm of the second summand on the right side is
    uniformly bounded in $\ell$. Furthermore, by assumption~\ref{ax:3} we have 
    $\varepsilon^{\gamma}-\tilde{\mu}^\gamma \leq p^2 +\nu$ for
    \begin{equation*}
      \nu = -\mu + \frac{6}{(2\pi)^{3/2}}
      \sup_{\ell> 0}\hat{V}_\ell(0) \|\hat{\gamma}_\ell \|_1,
    \end{equation*}
which is finite due to assumption~\ref{ax:positivity}  and Lemma~\ref{lemma:gamma_bound}. 
   In particular, 
    \begin{equation*}
      E_\mu^{\gamma_\ell,\Delta_\ell} \leq
      \sqrt{(p^2+\nu)^2+|\Delta_\ell|^2}.
    \end{equation*}
Recall that $\Delta_\ell(p)/\bar c_\ell$ converges pointwise to $P(p)$, and that $\bar c = \limsup_{\ell\to 0} \bar c_\ell$.  Assume, for the moment, that $\bar c < \infty$. Then, by dominated convergence,
    \begin{equation}
      \label{eq:lim_P}
      \limsup_{\ell\to 0}\int_{|p|\leq R}
      \frac{|\Delta_\ell|^2}{(p^2+\nu)^2+|\Delta_\ell|^2} \ddd{3}p
      = \int_{|p|\leq R}
      \frac{|\bar{c} P(p)|^2}{(p^2+\nu)^2+|\bar{c}P(p)|^2} \ddd{3}p
    \end{equation}
for any $R>0$. If $\bar c = \infty$, the same holds, with the integrand replaced by $1$. 
    In particular, if either $\bar c = \infty$ or $P$ is a polynomial of degree $n\geq 1$, the right side of (\ref{eq:lim_P}) diverges as $R\to \infty$, contradicting the uniform boundedness of $\Delta_\ell/E_\mu^{\gamma_\ell,\Delta_\ell}$  in $L^2(\R^3)$. We thus conclude that $n=0$ and $\bar{c} < \infty$, i.e., $\lim_{\ell\to 0} \Delta_\ell(p) = \bar c$ for a suitable subsequence. 

  (ii) The \label{lemma:convergence:int_gamma} uniform bound
    \eqref{eq:F_T_bound} for $\mathcal{F}_T^{V_\ell}$ implies that 
    $\hat{\gamma}_\ell$ is uniformly bounded in $L^2$. Thus, there
    is a subsequence which converges weakly to some $\hat{\gamma}$ in
    $L^2$. For that subsequence, we have for arbitrary $R>0$
    \begin{equation}
      \label{eq:convergence:Delta}
      \lim_{\ell\to 0}\int_{B_R(0)}\hat{\gamma}_\ell \ddd{3}p =
      \int_{B_R(0)}\hat{\gamma} \ddd{3}p.
    \end{equation}
   In particular, 
    \begin{equation*}
      \lim_{\ell\to 0} \int_{\mathbb{R}^3} \hat{\gamma}_\ell \ddd{3}p
      \geq \int_{\mathbb{R}^3} \hat{\gamma} \ddd{3}p.
    \end{equation*}
    Therefore, $\lim_{\ell\to 0} \int_{\mathbb{R}^3} \hat{\gamma}_\ell
    \ddd{3}p = \int_{\mathbb{R}^3} \hat{\gamma} \ddd{3}p + \delta$
    for an appropriate $\delta \geq 0$.  Then
    \begin{equation*}
      \lim_{\ell\to 0} \int_{|p| \geq R}
      \hat{\gamma}_\ell(p)|p|^b \ddd{3}p
      \geq
      R^b \lim_{\ell\to 0} \int_{|p| \geq R}
      \hat{\gamma}_\ell \ddd{3}p
      =
      R^b \lim_{\ell\to 0} \left[
        \int_{\mathbb{R}^3}
        \hat{\gamma}_\ell \ddd{3}p
        -
        \int_{|p| \leq R}
        \hat{\gamma}_\ell \ddd{3}p
      \right]
      \geq \delta R^b.
    \end{equation*}
    Since $R$ can be arbitrarily large and the left side is bounded, $\delta$ has to be $0$.

  (iii) This follows immediately from
    part (ii) together with assumption \ref{ax:positivity}. 
    \label{lemma:convergence:mu}
  
(iv)  Let \label{lemma:convergence:epsilon} $D_\ell(p) =
    \varepsilon^{\gamma_\ell}(p) - p^2$. We compute
    \begin{align}\nonumber
      |D_\ell(p)| &= 2 (2\pi)^{-3/2} |(\hat{V}_\ell-\hat{V}_\ell(0)) *
      \hat{\gamma}_\ell|
      \\ \nonumber 
      &\leq \frac 2 {(2\pi)^3} \int_{\mathbb{R}^3} \ddd{3}k \int_{\mathbb{R}^3}
      \ddd{3}x \left| V_\ell(x)(\ee^{-i (p-k)\cdot x } -1)
      \right| \hat{\gamma}_\ell(k) \\ \nonumber
      & \leq \frac{2 \|V_\ell\|_1}{(2\pi)^{3}} \int_{\mathbb{R}^3}
      \hat{\gamma}_\ell(k) \sup_{|x|\leq \ell}|\ee^{-i (p-k) \cdot x} -1| \ddd{3}k \\
      &\leq \frac{2 \|V_\ell\|_1}{(2\pi)^3} \ell^b\bigl(\bigl\|\hat{\gamma}_\ell
      |\cdot|^b\bigr\|_1 + \|\hat{\gamma}_\ell \|_1 |p|^b\bigr), \label{est:iv}
    \end{align}
    where we applied the fact that $|e^{it}-1|\leq |t|^b$ for
    $t\in\mathbb{R}$ and $0\leq b \leq 1$, as well as $|p-k|^b\leq
    |p|^b + |k|^b$. By Lemma~\ref{lemma:gamma_bound}, $\bigl\|\hat{\gamma}_\ell
    |\cdot|^b\bigr\|_1$ is uniformly bounded in $\ell$, hence this
    concludes the proof.

  (v)  \label{lemma:convergence:gamma} Recall the
    Euler-Lagrange equation \eqref{eq:el_gamma} for
    $\hat{\gamma}_\ell$, which states that
    \begin{equation*}
      \hat{\gamma}_\ell = \frac{1}{2} - \frac{\varepsilon^{\gamma_\ell}(p) -
        \tilde{\mu}^{\gamma_\ell}}{K_{T,\mu}^{\gamma_\ell,\Delta_\ell}(p)}.
    \end{equation*}
    We have just shown that the right side converges
    pointwise to
    \begin{equation}\label{v}
      \tilde{\gamma}(p) = \frac{1}{2} - \frac{p^2 -
        \tilde{\mu}^{\gamma}}{K_{T,\tilde\mu^\gamma}^{0,\Delta}(p)}.
    \end{equation}
    Since $\hat{\gamma}$ is the weak limit of $\hat{\gamma}_\ell$, it
    has to agree with the pointwise limit $\tilde{\gamma}$, i.e.,
    $\hat{\gamma}= \tilde{\gamma}$ almost everywhere.
       Therefore $\gamma$ satisfies Eq.~\eqref{eq:sat}.
  
(vi) We have already shown that the integrand converges pointwise. We want to use dominated convergence to show that also the integrals converge. For this purpose, we 
 rewrite the integrand in
    $m^{\gamma_\ell,\Delta_\ell}_\mu(T)$ in terms of $\gamma_\ell$.  With  $\xi(x) = \frac{x}{\ee^x-1}$, we have 
    \begin{equation*}
      \frac{1}{K_{T,\mu}^{\gamma_\ell,\Delta_\ell}} - \frac{1}{p^2} =
      \frac{p^2-K_{T,\mu}^{\gamma_\ell,\Delta_\ell}}{K_{T,\mu}^{\gamma_\ell,\Delta_\ell}
        p^2} = \frac{\varepsilon^{\gamma_\ell} -
        \tilde{\mu}^{\gamma_\ell}-E_\mu^{\gamma_\ell,\Delta_\ell}}{K_{T,\mu}^{\gamma_\ell,\Delta_\ell}
        p^2} + \frac{p^2 -
        (\varepsilon^{\gamma_\ell} -
        \tilde{\mu}^{\gamma_\ell})-2T \xi(\frac{E_\mu^{\gamma_\ell,\Delta_\ell}}{T}) }{K_{T,\mu}^{\gamma_\ell,\Delta_\ell} p^2}\,.
    \end{equation*}
   By comparing the first
    summand with the right side of \eqref{eq:el_gamma:2}, i.e.,
    \begin{equation*}
      \hat{\gamma}_\ell
      = \frac{1}{1 + \ee^{\frac{E_\mu^{\gamma_\ell,\Delta_\ell}}{T}}} +\frac{1}{2}
      \frac{E_\mu^{\gamma_\ell,\Delta_\ell}-(\varepsilon^{\gamma_\ell} -
        \tilde{\mu}^{\gamma_\ell})}
      {K_{T,\mu}^{\gamma_\ell,\Delta_\ell}},
    \end{equation*}
    we see that
    \begin{equation}
      \label{eq:K-p}
      \frac{1}{K_{T,\mu}^{\gamma_\ell,\Delta_\ell}} - \frac{1}{p^2}
      = -2\frac{\hat{\gamma_\ell}}{p^2}
      + \frac{2}{p^2}\frac{1}{1 + \ee^{\frac{E_\mu^{\gamma_\ell,\Delta_\ell}}{T}}}
      + \frac{p^2 -
        (\varepsilon^{\gamma_\ell} -
        \tilde{\mu}^{\gamma_\ell})-2T \xi(\frac{E_\mu^{\gamma_\ell,\Delta_\ell}}{T}) }{K_{T,\mu}^{\gamma_\ell,\Delta_\ell} p^2}\,.
    \end{equation}
    We can now argue as above to show that, by dominated convergence, the
    integrals of all summands on the right side except for $-2
    \frac{\hat{\gamma}_\ell}{p^2}$ converge to their corresponding
    expressions with $\gamma_\ell$ replaced by its limit $\gamma$ and
    $\Delta_\ell$ replaced by $\Delta$.  Indeed, assumption
    \ref{ax:3} implies $\varepsilon^{\gamma_\ell} -
    \tilde{\mu}^{\gamma_\ell} \geq p^2 -\mu$ and thus
    \begin{equation*}
      E_\mu^{\gamma_\ell,\Delta_\ell}
      \geq |\varepsilon^{\gamma_\ell} - \tilde{\mu}^{\gamma_\ell}|
      \geq \varepsilon^{\gamma_\ell} - \tilde{\mu}^{\gamma_\ell}
      \geq p^2 -\mu.
    \end{equation*}
    For this reason,
    \begin{equation*}
      \frac{2}{p^2}\frac{1}{1 + \ee^{\frac{E_\mu^{\gamma,\Delta}}{T}}}
      \leq \frac{2}{p^2}\frac{1}{1 + \ee^{\frac{p^2-\mu}{T}}}.
    \end{equation*}
    Moreover, the function
    \begin{equation}\label{def:kappac}
      \kappa_c(x) =
      \begin{cases}
        \frac{x}{\tanh(x)},& x\geq 0\\ 1,& x\leq 0
      \end{cases}
    \end{equation}
    is monotone increasing, so
    \begin{equation*}
      K_{T,\mu}^{\gamma_\ell,\Delta_\ell}
      = 2T \kappa_c\Big(\frac{E_\mu^{\gamma_\ell,\Delta_\ell}}{2T}\Big)
      \geq 2T \kappa_c\Big(\frac{p^2-\mu}{2T}\Big).
    \end{equation*}
    Together with $\xi(x)\leq 1$ for $x\geq 0$ and the bound (\ref{est:iv}) on $|p^2-\varepsilon^{\gamma_\ell}(p)|$, this implies the statement. 

Finally, we can argue as in (ii) above to conclude that  $\lim_{\ell\to 0}\int_{\mathbb{R}^3}
    \frac{\hat{\gamma}_\ell}{p^2} \ddd{3}p = \int_{\mathbb{R}^3}
    \frac{\hat{\gamma}}{p^2} \ddd{3}p$, and hence obtain the desired result
    \begin{equation*}
      \lim_{\ell\to 0} \int_{\mathbb{R}^3} \left(\frac{1}{K_{T,\mu}^{\gamma_\ell,\Delta_\ell}}
        - \frac{1}{p^2}\right) \ddd{3} p
      = \int_{\mathbb{R}^3} \left(\frac{1}{K_{T,\mu}^{\gamma,\Delta}}
        - \frac{1}{p^2}\right) \ddd{3} p,
    \end{equation*}
    where we used that the limit
    $\gamma$ also satisfies a suitable Euler-Lagrange equation, as shown in (v), and hence satisfies an identity as in
    \eqref{eq:K-p} as well.  
\end{proof}

With the aid of Lemma~\ref{lemma:convergence}, we can now give the 

\begin{proof}[Proof of Theorem \ref{thm:gap_eff}]
  The convergence of $|\Delta_\ell(p)|$, $\tilde\mu^{\gamma_\ell}$ and
  $\hat\gamma_\ell(p)$ follows immediately from
  Lemma~\ref{lemma:convergence}, at least for a suitable
  subsequence. To prove the validity of \eqref{eq:gap_eff}, we follow
  a similar strategy as in \cite[Lemma 1]{HS-mu}.  From Theorem
  \ref{thm:stability} we know that
  \begin{equation*}
    (K_{T,\mu}^{\gamma_\ell,\Delta_\ell} + V_\ell) {\alpha}_\ell =
    0, \quad
    \textrm{ with } {\alpha}_\ell \in H^1(\mathbb{R}^3)\,,
  \end{equation*}
  and we assume that $\alpha_\ell$ is not identically zero. According  to the Birman-Schwinger principle,
  $K_{T,\mu}^{\gamma_\ell,\Delta_\ell} + V_\ell$ has $0$ as eigenvalue
  if and only if
  \begin{equation*}
    V_\ell^{1/2}\frac{1}{K_{T,\mu}^{\gamma_\ell,\Delta_\ell}} |V_\ell|^{1/2}
  \end{equation*}
  has $-1$ as an eigenvalue.
  
  We decompose
  $V_\ell^{1/2}\frac{1}{K_{T,\mu}^{\gamma_\ell,\Delta_\ell}}
  |V_\ell|^{1/2}$ as
  \begin{align*}
    V_\ell^{1/2} \frac{1}{K_{T,\mu}^{\gamma_\ell,\Delta_\ell}}
    |V_\ell|^{1/2} = V_\ell^{1/2} \frac{1}{p^2} |V_\ell|^{1/2} +
    m_\mu^{\gamma_\ell,\Delta_\ell}(T) |V_\ell^{1/2}\rangle \langle
    |V_\ell|^{1/2}| + A_{\mu,T,\ell},
  \end{align*}
  where
  \begin{align*}
    A_{\mu,T,\ell} &= V_\ell^{1/2}
    \Big(\frac{1}{K_{T,\mu}^{\gamma_\ell,\Delta_\ell}} -
    \frac{1}{p^2}\Big) |V_\ell|^{1/2} -
    m_\mu^{\gamma_\ell,\Delta_\ell}(T) |V_\ell^{1/2}\rangle \langle
    |V_\ell|^{1/2}|.
  \end{align*}
  
  By assumption \ref{ax:A}, $1+V_\ell^{1/2} \frac{1}{p^2}
  |V_\ell|^{1/2}$ is invertible. Hence we can write 
  \begin{align*}
    1 + V_\ell^{1/2}\frac{1}{K_{T,\mu}^{\gamma_\ell,\Delta_\ell}}
    |V_\ell|^{1/2}
    = &\left(1+V_\ell^{1/2} \frac{1}{p^2} |V_\ell|^{1/2} \right) \times\\
    &\times\left(1+\frac{1}{1+ V_\ell^{1/2} \frac{1}{p^2}
        |V_\ell|^{1/2}}
      \left(m_\mu^{\gamma_\ell,\Delta_\ell}(T)|V_\ell^{1/2}\rangle
        \langle |V_\ell|^{1/2}| + A_{T,\mu,\ell}\right)\right)\,,
  \end{align*}
  and conclude that the operator
  \begin{equation*}
    \frac{1}{1+ V_\ell^{1/2} \frac{1}{p^2}
      |V_\ell|^{1/2}}\left(m_\mu^{\gamma_\ell,\Delta_\ell}(T)
      |V_\ell^{1/2}\rangle \langle |V_\ell|^{1/2}| +
      A_{T,\mu,\ell}\right)
  \end{equation*}
  has an eigenvalue $-1$.

  We are going to show below that
  \begin{equation}
    \label{eq:A:lim}
    \lim_{\ell\to 0} \left\| \frac{1}{1+ V_\ell^{1/2}
        \frac{1}{p^2}|V_\ell|^{1/2}} A_{\mu,T,\ell}\right\| = 0.
  \end{equation}
As a consequence, $1+(1+V_\ell^{1/2}
        \frac{1}{p^2}|V_\ell|^{1/2})^{-1} A_{\mu,T,\ell}$ is invertible for small $\ell$, and we can argue as above to conclude that the rank one operator
$$
m_\mu^{\gamma_\ell,\Delta_\ell}(T) \left(1+  \frac{1}{1+ V_\ell^{1/2}
        \frac{1}{p^2}|V_\ell|^{1/2}} A_{\mu,T,\ell}\right)^{-1}   \frac{1}{1+ V_\ell^{1/2} \frac{1}{p^2}
      |V_\ell|^{1/2}}|V_\ell^{1/2}\rangle \langle |V_\ell|^{1/2}|
$$
has an eigenvalue $-1$, i.e.,
\begin{equation}\label{trw}
-1 = m_\mu^{\gamma_\ell,\Delta_\ell}(T)\left\langle |V_\ell|^{1/2}\left| \left(1+  \frac{1}{1+ V_\ell^{1/2}
        \frac{1}{p^2}|V_\ell|^{1/2}} A_{\mu,T,\ell}\right)^{-1}   \frac{1}{1+ V_\ell^{1/2} \frac{1}{p^2}
      |V_\ell|^{1/2}}\right|V_\ell^{1/2}\right\rangle \,.
\end{equation}
With the aid of \eqref{eq:a} and the resolvent identity, we can 
rewrite \eqref{trw} as 
\begin{align}\nonumber
& 4\pi a(V_\ell) +  \frac 1 { m_\mu^{\gamma_\ell,\Delta_\ell}(T)} \\ &= \left\langle |V_\ell|^{1/2}\left|  \frac{1}{1+ V_\ell^{1/2}
        \frac{1}{p^2}|V_\ell|^{1/2}} A_{\mu,T,\ell} \left(1+  \frac{1}{1+ V_\ell^{1/2}
        \frac{1}{p^2}|V_\ell|^{1/2}} A_{\mu,T,\ell}\right)^{-1}   \frac{1}{1+ V_\ell^{1/2} \frac{1}{p^2}
      |V_\ell|^{1/2}}\right|V_\ell^{1/2}\right\rangle \,. \label{aab}
\end{align}

We are going to show below that the term on the right side of (\ref{aab}) goes to zero as $\ell\to 0$ and, as a consequence, 
  \begin{equation}\label{aac}
    \lim_{\ell\to 0} m_\mu^{\gamma_\ell,\Delta_\ell}(T) = -\lim_{\ell\to 0}  \frac{1}{4\pi a(V_\ell)} =  -\frac{1}{4\pi a} \,.
  \end{equation}
  On the other hand, by Lemma~\ref{lemma:convergence}
  there is a subsequence  of $(\gamma_\ell,\alpha_\ell)$ such that
  \begin{equation*}
    \lim_{\ell\to 0} m_\mu^{\gamma_\ell,\Delta_\ell}(T) = 
    \frac{1}{(2\pi)^3}\int_{\mathbb{R}^3}\left(\frac{1}{K_{T,\tilde{\mu}}^{0,\Delta}}
      -\frac{1}{p^2}\right) \ddd{3}p,
  \end{equation*}
  where $\Delta$ is the pointwise limit of $|\Delta_\ell(p)|$
  and
  $\tilde{\mu}$ is the limit of $\tilde{\mu}^{\gamma_\ell}$. This
  shows \eqref{eq:gap_eff}, at least for a subsequence. 

 It remains to show \eqref{eq:A:lim} and \eqref{aac}.
  We start with the decomposition
  \begin{equation}
    \label{eq:B:decomposition}
    \frac{1}{1+ V_\ell^{1/2}
      \frac{1}{p^2}|V_\ell|^{1/2}} = \frac{1}{e_\ell} P_\ell + \frac{1}{1+ V_\ell^{1/2}
      \frac{1}{p^2}|V_\ell|^{1/2}}(1-P_\ell),
  \end{equation}
  where the second summand is uniformly bounded by assumption
  \ref{ax:A}. The integral kernel of $A_{\mu,T,\ell}$ is given by
 \begin{equation}\label{axy}
    A_{\mu,T,\ell}(x,y) =
    \frac{V_\ell(x)^{\frac{1}{2}}|V_\ell(y)|^{\frac{1}{2}}}
    {(2\pi)^3}\int_{\mathbb{R}^3}\left(
      \frac{1}{K_{T,\mu}^{\gamma_\ell,\Delta_\ell}}-\frac{1}{p^2}\right)
    (\ee^{-i(x-y)\cdot p}-1) \ddd{3} p\,,
 \end{equation}
which can  be estimated as
  \begin{equation}
    \label{eq:A}
    |A_{\mu,T,\ell}(x,y)|
    \leq
    \frac{|V_\ell(x)|^{\frac{1}{2}}|V_\ell(y)|^{\frac{1}{2}}}
    {(2\pi)^3}\int_{\mathbb{R}^3}\left|
      \frac{1}{K_{T,\mu}^{\gamma_\ell,\Delta_\ell}}-\frac{1}{p^2}\right|
    (|x-y|\ |p|)^q \ddd{3} p
  \end{equation}
  for any $0\leq q \leq 1$.  In the proof of
  Lemma~\ref{lemma:convergence} \ref{lemma:convergence:m} we found
  that
 the integral
  \begin{equation*}
    \int_{\mathbb{R}^3}\left|
      \frac{1}{K_{T,\mu}^{\gamma_\ell,\Delta_\ell}}-\frac{1}{p^2}\right|
    |p|^q \ddd{3} p
  \end{equation*}
  is uniformly bounded in $\ell$ for $q<1$. With the aid of Assumption \ref{ax:1}, we can thus bound the Hilbert-Schmidt norm of $A_{\mu,T,\ell}$ as 
  \begin{equation}
    \|A_{\mu,T,\ell}\|_2 \leq \textrm{const}\, \ell^q \|V_\ell\|_1 \,. \label{eq:A_q_bound}
  \end{equation}
  In particular, because of Assumption \ref{ax:L1}, 
  \begin{equation*}
    \left\|\frac{1}{1+ V_\ell^{1/2}
      \frac{1}{p^2}|V_\ell|^{1/2}}(1-P_\ell)
    A_{\mu,T,\ell}\right\| \leq O(\ell^q)
  \end{equation*}
  for small $\ell$. It remains to show that the
  contribution of the first summand in \eqref{eq:B:decomposition} to the norm in question 
  vanishes as well.  We have
$$
\|P_\ell A_{\mu,T,\ell}\| = \frac{ \|A^*_{\mu,T,\ell} J_\ell \phi_\ell\|}{|\langle J_\ell \phi_\ell| \phi_\ell\rangle|}\,.
$$
By \eqref{eq:A},
$$
\left| \left( A^*_{\mu,T,\ell} J_\ell \phi_\ell\right)(x)\right|  \leq C \ell^q |V(x)|^{1/2} \int_{\R^3} |V(y)|^{1/2}|\phi_\ell(y)|  \ddd{3}y\,,
$$
and hence 
  \begin{equation}\label{usi}
    \| P_\ell A_{\mu,T,\ell} \|  \leq \textrm{const}\ \ell^q \frac{1}{|\langle J_\ell
      \phi_\ell|\phi_\ell\rangle|} \left\langle\left.
    |V_\ell|^{1/2}\right| |\phi_\ell|\right\rangle \|V_\ell\|_1^{1/2}.
  \end{equation}
  By \ref{ax:n}, we know that $\frac{\langle
  |V_\ell|^{1/2}||\phi_\ell|\rangle}{|\langle J_\ell
    \phi_\ell|\phi_\ell\rangle|} \leq O(\ell^{1/2})$. Since $e_\ell = O(\ell)$ by assumption, we arrive at
  \begin{equation*}
    \biggl\|\frac{1}{1+ V_\ell^{1/2} \frac{1}{p^2}|V_\ell|^{1/2}}
    A_{\mu,T,\ell}\biggr\| \leq O(\ell^{q-1/2}),
  \end{equation*}
  which vanishes by choosing $1/2 < q <1$.

To  show \eqref{aac}, i.e., that the term on the right side of \eqref{aab} vanishes as $\ell\to 0$, we can again use the decomposition \eqref{eq:B:decomposition} to argue that
\begin{equation}\label{asin}
\left\| \left.\left. \left(1+  \frac{1}{1+ V_\ell^{1/2}
        \frac{1}{p^2}|V_\ell|^{1/2}} A_{\mu,T,\ell}\right)^{-1}   \frac{1}{1+ V_\ell^{1/2} \frac{1}{p^2}
      |V_\ell|^{1/2}}\right|V_\ell^{1/2}\right\rangle \right\| \leq O(\ell^{-1/2}) \,,
\end{equation}
where we used (\ref{eq:A:lim}) as well as Assumptions \ref{ax:L1}, \ref{ax:A} and \ref{ax:n}. Moreover, 
\begin{equation}\label{asin2}
\left\|\left.\left. A_{\mu,T,\ell}^* \frac{1}{1+ |V_\ell|^{1/2}
        \frac{1}{p^2}V_\ell^{1/2}} \right|  |V_\ell|^{1/2} \right\rangle \right\| 
    \leq O(\ell^q) + \frac 1{e_\ell} \left\|P_\ell A_{\mu,T,\ell}\right\| \left\| P_\ell^*  \left| |V_\ell|^{1/2}\right\rangle\right\| \leq O(\ell^{q })
\end{equation}
using (\ref{usi}). The last term in (\ref{aab}) thus is of order $\ell^{q-1/2}$, and vanishes as $\ell\to 0$ for any $1/2 < q <1$.  This proves (\ref{aac}).

  As a last step, we show
  that the limit points for
  $\tilde{\mu}^{{\gamma}_\ell}$ and $|\Delta_\ell(p)|$, and thus also of $\hat\gamma_\ell(p)$,  are unique.
  We use the fact that the limit points solve the two implicit
  equations \eqref{eq:sat} and \eqref{eq:gap_eff}, i.e.,
  \begin{align*}
    F(\tilde{\mu},\Delta) = 0\, ,\qquad G(\tilde{\mu}, \Delta) = 0\, ,
  \end{align*}
  where
  \begin{align*}
    F(\nu,\Delta) &= \nu-\mu + \frac{\mathcal{V}}{(2\pi)^{3/2}}
    \int_{\mathbb{R}^3} \left(1 - 
      \frac{p^2-\nu}{ K_{T,\nu}^{0,\Delta}}\right)\ddd{3}p\,,\\
    G(\nu, \Delta) &=  \frac{1}{4\pi
      a} + 
  \frac{1}{(2\pi)^3}\int_{\mathbb{R}^3}\left(\frac{1}{K_{T,\nu}^{0,\Delta}}
  -\frac{1}{p^2}\right) \ddd{3}p\,.
  \end{align*}
  It is straightforward to check that 
  \begin{align*}
    \partial_\nu F &> 0 & \partial_\nu G &> 0 \\
    \partial_\Delta F &> 0 & \partial_\Delta G &< 0
  \end{align*}
(compare with similar computations in Appendix~\ref{sec:appendix:tc}).
Hence the set where $F$ vanishes defines a strictly decreasing curve $\R_+\to \R$, while the analogous curve for the zero-set of $G$ is strictly increasing. Consequently, they can intersect at most once. 

This proves uniqueness under the assumptions that $\Delta_\ell\neq 0$ for a sequence of $\ell$'s going to zero. In the opposite case, $\Delta_\ell = 0$ for $\ell$ small enough, hence $\Delta = 0$. The uniqueness in this case follows as above, looking at the equation $F(\tilde\mu,0)=0$. This completes the proof of Theorem~\ref{thm:gap_eff}.
\end{proof}

\begin{remark}
  In case $K_{T,\mu}^{\gamma_\ell,\Delta_\ell}$ is
  reflection-symmetric in $p$, one can show that the bound \eqref{eq:A_q_bound} holds also for $q=1$. 
  Indeed, in this case only the symmetric part of $\ee^{-i(x-y)\cdot
    p}-1$ contributes to the integral kernel
  of $A_{\mu,T,\ell}$, and hence 
  \begin{equation*}
    A_{\mu,T,\ell}(x,y) = \frac{V_\ell^{\frac{1}{2}}(x)|V_\ell|^{\frac{1}{2}}(y)}
    {(2\pi)^3}\int_{\mathbb{R}^3}\left(
      \frac{1}{K_{T,\mu}^{\gamma_\ell,\Delta_\ell}}-\frac{1}{p^2}\right)
    \Big(\cos\big((x-y)\cdot p\big)-1\Big) \ddd{3} p\,.
  \end{equation*}
  Again using \eqref{eq:K-p}, we may write
  \begin{equation*}
    \left|\frac{1}{K_{T,\mu}^{\gamma_\ell,\Delta_\ell}}-\frac{1}{p^2}\right|
    \leq {\rm const}\frac{1}{1+p^4} + R_\ell(p),
  \end{equation*}
  such that
  \begin{equation*}
    \int_{\mathbb{R}^3}\left|p R_\ell(p)\right|
    \ddd{3} p
  \end{equation*}
  is uniformly bounded in $\ell$.  Since
  \begin{equation*}
    \int_{\mathbb{R}^3} \frac{1-\cos\big(p\cdot (x-y)\big)}{1+p^4} \ddd{3}p
    =
    \sqrt{2}\pi^2\left[ 1 -
      \ee^{-\frac{|x-y|}{\sqrt{2}}}\frac{\sin\big(\frac{|x-y|}{\sqrt{2}}\big)}{|x-y|/\sqrt{2}}\right]
    \leq \pi^2 |x-y|,
  \end{equation*}
  we get
  \begin{equation*}
    \|A_{\mu,T,\ell}\|_2 \leq \textrm{const} \left[\int_{\mathbb{R}^3}
      |V_\ell(x)||V_\ell(y)||x-y|^{2}\ddd{3}x\ddd{3}y\right]^{1/2} \leq O(\ell)
  \end{equation*}
in this case.
\end{remark}

\subsection{Critical Temperature}

In this section we will prove Theorem~\ref{thm:T_c}. We start with the following observation.

\begin{lemma} \label{thm:m_monotony}
 Let $\mu > 0$, $T<T_c$, and let $(\gamma^0_\ell,0)$ be a family of normal states for $\mathcal{F}_T^{V_\ell}$. Then
  \begin{equation}
 \liminf_{\ell\to 0}    m^{\gamma^0_\ell,0}_\mu(T) > - \frac 1{4\pi a} \,.
\end{equation}
\end{lemma}

\begin{proof}
By mimicking the proof of Lemma~\ref{lemma:convergence}, we observe that (for a suitable subsequence)
$$
\lim_{\ell \to 0} m^{\gamma^0_\ell,0}_\mu(T) = m_{\tilde\mu^\gamma}^{0,0}(T)
$$
where $\tilde \mu^\gamma = \mu - 2(2\pi)^{-3/2} \mathcal{V}
\int_{\R^3} \hat\gamma(p) \ddd{3}p$ and $\hat\gamma(p) =
(1+\ee^{(p^2-\tilde\mu^\gamma)/T})^{-1}$. It is shown in
Appendix~\ref{sec:appendix:tc} that $m_{\tilde\mu^\gamma}^{0,0}(T)$ is
a strictly decreasing function of $T$. At $T=T_c$, it equals $-1/(4\pi
a)$ according to Definition~\ref{def:tc}, hence
$m_{\tilde\mu^\gamma}^{0,0}(T)> -1/(4\pi a)$ for $T<T_c$.
\end{proof}

The first part of Theorem~\ref{thm:T_c} then follows from the following lemma.

\begin{lemma}
  \label{lemma:short_range:suf}
  Let $(\gamma^0_\ell,0)$ be a normal state of
  $\mathcal{F}_T^{V_\ell}$.  Assume that $\lim_{\ell\to 0}
  m_\mu^{\gamma^0_\ell,0}(T) > -\frac{1}{4\pi a}$. Then, for small
  enough $\ell$, the linear operator $K_{T,\mu}^{\gamma^0_\ell,0} +
  V_\ell$ has at least one negative eigenvalue.
\end{lemma}

\begin{proof}
  With the aid of the Birman-Schwinger principle, we will attribute the
  existence of an eigenvalue of $K_{T,\mu}^{\gamma^0_\ell,0}
  + V_\ell$ below the essential spectrum
  to a solution of a certain implicit equation. We then show
  the existence of such a solution, which proves the existence of a negative
  eigenvalue.

  Note that the infimum of the essential spectrum of  $K_{T,\mu}^{\gamma^0_\ell,0}
  + V_\ell$ is $2T$. Let $e < 2T$. According to the Birman-Schwinger principle,
  $K_{T,\mu}^{\gamma^0_\ell,0} + V_\ell$ has an eigenvalue
  $e$ if and only if
  \begin{equation}\label{op}
    V_\ell^{1/2}\frac{1}{K_{T,\mu}^{\gamma^0_\ell,0}-e} |V_\ell|^{1/2}
  \end{equation}
  has an eigenvalue $-1$.  As in the proof of Theorem
  \ref{thm:gap_eff}, we decompose the operator (\ref{op}) 
  as
  \begin{equation*}
    V_\ell^{1/2} \frac{1}{K_{T,\mu}^{\gamma^0_\ell,0}-e} |V_\ell|^{1/2}
    = V_\ell^{1/2} \frac{1}{p^2} |V_\ell|^{1/2}
    + m_{\mu,e}^{\gamma^0_\ell}(T) |V_\ell^{1/2}\rangle \langle |V_\ell|^{1/2}| + A_{\mu,T,\ell,e}
  \end{equation*}
  where
  \begin{equation}
    \label{eq:short_range:mu}
    m_{\mu,e}^{\gamma^0_\ell}(T) =
    \frac{1}{(2\pi)^3}\int_{\mathbb{R}^3}\left(\frac{1}{K_{T,\mu}^{\gamma^0_\ell,0}-e}
      -\frac{1}{p^2}\right) \ddd{3}p \,.
  \end{equation}
  We claim that the remainder $A_{\mu,T,\ell,e}$ is bounded above by $O(\ell^q)$ in Hilbert-Schmidt norm, for any $0\leq  q<1$, uniformly in $e$ for $e\leq 0$.  
  This will follow from the same 
  estimates as in the proof of Theorem \ref{thm:gap_eff} if we can show that 
  \begin{equation*}
    \int_{\mathbb{R}^3}|p|^q
    \left(\frac{1}{K_{T,\mu}^{\gamma^0_\ell,0}-e}-\frac{1}{K_{T,\mu}^{\gamma^0_\ell,0}}\right) \ddd{3}p
  \end{equation*}
is uniformly bounded in $\ell$ for $0\leq q<1$. 
  But since
  \begin{equation*}
    \frac{1}{K_{T,\mu}^{\gamma^0_\ell,0}-e}-\frac{1}{K_{T,\mu}^{\gamma^0_\ell,0}}
    = \frac{e}{K_{T,\mu}^{\gamma^0_\ell,0}(K_{T,\mu}^{\gamma^0_\ell,0}-e)}
    \leq
    \frac{e}{(2T)^2}\frac{1}{\kappa_c(\frac{p^2-\mu}{2T})\big(\kappa_c(\frac{p^2-\mu}{2T})-\frac{e}{2T}\big)}\,,
  \end{equation*}
with $\kappa_c$ defined in (\ref{def:kappac}), 
this is indeed the case.

  Again, the operator
  \begin{equation*}
    1+V_\ell^{1/2}\frac{1}{p^2}|V_\ell|^{1/2} + A_{\mu,T,\ell,e}
    =
    \Bigl(1+V_\ell^{1/2}\frac{1}{p^2}|V_\ell|^{1/2}\Bigr)\biggl(1+\frac{1}{1+V_\ell^{1/2}\frac{1}{p^2}|V_\ell|^{1/2}}A_{\mu,T,\ell,e}\biggr)
  \end{equation*}
  is invertible for small $\ell$, by assumption \ref{ax:A}
  and the fact that 
  \begin{equation}\label{asii}
    \lim_{\ell \to 0}  \biggl\|
      \frac{1}{1+V_\ell^{1/2}\frac{1}{p^2}|V_\ell|^{1/2}}A_{\mu,T,\ell,e}\biggr\|=0\,,
  \end{equation}
  with the same argument as in the proof of Theorem \ref{thm:gap_eff}.
  We conclude that $K_{T,\mu}^{\gamma^0_\ell,0} + V_\ell$
  has an eigenvalue $e$ if and only if the rank one operator
  \begin{equation*}
     m_{\mu,e}^{\gamma^0_\ell}(T) \left(1+V_\ell^{1/2}\frac{1}{p^2}|V_\ell|^{1/2} +
    A_{\mu,T,\ell,e}\right)^{-1} |V_\ell^{1/2}\rangle
    \langle |V_\ell|^{1/2}|
  \end{equation*}
  has an eigenvalue $-1$, i.e., if
  \begin{equation}
    \label{eq:eigenvalue}
    \tilde{a}_{\ell,e} = \left\langle |V_\ell|^{1/2}\left|\frac{1}{1+V_\ell^{1/2}\frac{1}{p^2}|V_\ell|^{1/2} +
      A_{\mu,T,\ell,e}} \right|V_\ell^{1/2}\right\rangle = -\frac{1}{m_{\mu,e}^{\gamma^0_\ell}(T)}.
  \end{equation}

  We claim that, for small enough $\ell$, the implicit equation
  \eqref{eq:eigenvalue} has a solution $e < 0$, which implies the
  existence of a negative eigenvalue of
  $K_{T,\mu}^{\gamma^0_\ell,0} + V_\ell$.  
We first argue that  $\lim_{\ell\to 0} \tilde{a}_{\ell,e} = 4\pi a$. This follows from the same arguments as in (\ref{asin})--(\ref{asin2}), in fact. Recall that,  by
  assumption, $\lim_{\ell\to 0} m_\mu^{\gamma^0_\ell,0}(T) >
  -\frac{1}{4\pi a}$.  Moreover the integral
  $m_{\mu,e}^{\gamma^0_\ell}(T)$ is monotone increasing in $e$ and $e
  \mapsto m_{\mu,e}^{\gamma^0_\ell}(T)$ maps $(-\infty,0]$ onto the
  interval $(-\infty, m_\mu^{\gamma^0_\ell,0}(T)]$.  Since
  $\tilde{a}_{\ell,e}$ depends continuously on $e$, there has to be a solution $e < 0$ to \eqref{eq:eigenvalue} for small enough
  $\ell$, and  thus $K_{T,\mu}^{\gamma^0_\ell,0} + V_\ell$ must have a
  negative eigenvalue. This completes the proof.
\end{proof}

We now give the 

\begin{proof}[Proof of Theorem~\ref{thm:T_c}]
  Part (i) follows immediately from Lemmas~\ref{thm:m_monotony}
  and~\ref{lemma:short_range:suf}. To prove part (ii), we argue by
  contradiction. Suppose that $T>T_c$ and that there does not exist an
  $\ell_0(T)$ such that for $\ell<\ell_0(T)$ all minimizers of
  $\mathcal{F}_T^{V_\ell}$ are normal. Then there exists a sequence of
  $\ell$'s going to zero and corresponding minimizers
  $(\gamma_\ell,\alpha_\ell)$ with $\alpha_\ell\neq 0$ and thus, by
  Theorem~\ref{thm:gap_eff}, Eqs.~\eqref{eq:sat}
  and~\eqref{eq:gap_eff} hold in the limit $\ell \to 0$. We claim that
  these equations do not have a solution for $T>T_c$, thus providing
  the desired contradiction.

At the end of the proof of Theorem~\ref{thm:gap_eff}, we have already argued that the right side of (\ref{eq:gap_eff}) is monotone decreasing in $\Delta$ and increasing in $\tilde \mu$. Moreover, $\tilde \mu$ is decreasing in $\Delta$. In particular, we conclude from Eqs.~\eqref{eq:sat}
  and~\eqref{eq:gap_eff} that
$$
      -\frac{1}{4\pi a} \leq  \frac{1}{(2\pi)^3}\int_{\mathbb{R}^3}
      \left(
        \frac{\tanh\big(\frac{p^2-\tilde{\mu}}{2T}\big)}{p^2-\tilde{\mu}}
        -\frac{1}{p^2} \right) \ddd{3}p
$$
with $\tilde \mu$ given by
$$
      \tilde{\mu} = \mu -\frac{2
        \mathcal{V}}{(2\pi)^{3/2}}\int_{\mathbb{R}^3}
      \frac{1}{1+\ee^{\frac{p^2-\tilde{\mu}}{T}}} \ddd{3}p\,.
$$
According to our analysis in Appendix~\ref{sec:appendix:tc}, this implies $T\leq T_c$, however.
\end{proof}

\section*{Acknowledgments}
We would like to thank Max Lein and Andreas Deuchert for valuable
suggestions and remarks. Partial financial support by the NSERC (R.S.) is gratefully acknowledged.

\appendix

\section{Example for a Sequence of Short-Range Potentials}
\label{sec:example_scattering_length}

In dimension $d=3$, contact potentials are realized by a one-parameter
family $-\Delta_a$ of self-adjoint extensions of the Laplacian
$\left. -\Delta\right|_{C_0^\infty(\mathbb{R}^3\setminus \{0\})}$,
indexed by the scattering length $a$.  Moreover, $-\Delta_a$ can be
obtained as a norm resolvent limit of short-range Hamiltonians of the
form $-\Delta + V_\ell$. This is presented in \cite{albeverio82, albeverio} in the case of $0 < \lim_{\ell\to 0}\|V_\ell\|_{3/2} <
\infty$, and was extended in \cite{BHS-delta} to cases where $0 < \lim_{\ell\to 0}\|V_\ell\|_{1} <
\infty$. In this Appendix, we use an approach similar to \cite{BHS-delta} to construct a sequence of
potentials $V_\ell$ converging to a contact potential.  In particular,
we are interested in the case where the scattering length $a(V_\ell)$
converges to a negative value $a < 0$, and where all the assumptions in 
Assumption~\ref{asm:assumption} are satisfied.

\subsection{Example 1}
\label{sec:ex:1}

As a first example, we follow \cite[chap I.1.2-4]{albeverio}.  We
start with an arbitrary potential $V \in
L^1(\mathbb{R}^3) \cap L^2(\mathbb{R}^3)$, such that
\begin{enumerate}
\item $p^2 + V(x)\geq 0$, and $V$ has a simple zero-energy resonance, i.e., there is a simple
  eigenvector $\phi \in L^2(\mathbb{R}^3)$ with
  $(V^{1/2}\frac{1}{p^2}|V|^{1/2} +1) \phi = 0$, and $\psi(x) =
  \frac{1}{p^2}|V|^{1/2} \phi \in L^2_{\textrm{loc}}(\mathbb{R}^3)$,
  \label{ex:1:ax:1}
\item $\|\hat{V}(p)\| \leq 2 \hat{V}(0)$.
\end{enumerate}
Define $V_\ell(x) = \lambda(\ell) \ell^{-2} V(\frac{x}{\ell})$, where
$\lambda(0) = 1$, $\lambda < 1$ for all $\ell>0$ and
$1-\lambda(\ell) = O(\ell)$.  
The important point of this scaling is the following.
Denote by $U_\ell$ the unitary
scaling operator $(U_\ell\varphi)(x) =
\ell^{-3/2}\varphi(\frac{x}{\ell})$. 
By a simple calculation one obtains the relation
$$ U_\ell V^{1/2} \frac 1 {p^2} |V|^{1/2} U_\ell^{-1} = \frac{1}{\lambda(\ell)} V_\ell^{1/2} \frac 1 {p^2} |V_\ell|^{1/2},$$
such that, with $\phi_\ell = U_\ell \phi$,
\begin{equation}\label{Xell}  V_\ell^{1/2} \frac 1 {p^2} |V_\ell|^{1/2} \phi_\ell = \lambda(\ell) U_\ell V^{1/2} \frac 1 {p^2} |V|^{1/2} \phi = - \lambda(\ell) \phi_\ell.
\end{equation}
This shows that the lowest eigenvalue of 
$1 + V_\ell^{1/2} \frac 1 {p^2} |V_\ell|^{1/2}$ is $1 - \lambda(\ell) = O(\ell)$.  

Moreover, by
construction, $\hat{V}_\ell(p) = \ell \lambda(\ell)\hat{V}(\ell p)$,
$\|V_\ell\|_{3/2} = \lambda(\ell)\|V\|_{3/2}$, and the 1-norm can be bounded as $\|V_\ell\|_1 \leq
\left(\frac{4}{3}\pi\right)^{1/3} \ell \lambda(\ell)\|V\|_{3/2}$, 
hence \ref{ax:1}, \ref{ax:supp} and \ref{ax:L1}--\ref{ax:L2} hold.

The validity of Assumption {\ref{ax:infspec}} is a consequence of the following general fact.

\begin{lemma}
  \label{lemma:infspec}
 If $\|V_\ell\|_{3/2}$ is uniformly
  bounded,  assumptions \ref{ax:1} and \ref{ax:A} imply assumption \ref{ax:infspec}.
\end{lemma}

\begin{proof}
  We look for  $C>0$ such that $p^2+V_\ell-|p|^b + C$ is non-negative for all $\ell>0$. By the Birman-Schwinger principle, this is the case
  if and only if 
  \begin{equation*}
    1 +V_\ell^{1/2}\frac{1}{p^2-|p|^b+C+E}|V_\ell|^{1/2} = 1 + J_\ell X_\ell^{C+E} +R_\ell^E
  \end{equation*}
  is invertible for all $E>0$. Here $J_\ell = \big\{
  \begin{smallmatrix}
    1,& V_\ell \geq 0\\
    -1,& V_\ell <0
  \end{smallmatrix}$, $X_\ell^E =
  |V_\ell|^{1/2}\frac{1}{p^2+E}|V_\ell|^{1/2}$ and
  \begin{equation*}
    R_\ell^E =
    V_\ell^{1/2}\frac{1}{p^2-|p|^b+C+E}|V_\ell|^{1/2} -
    V_\ell^{1/2}\frac{1}{p^2+C+E}|V_\ell|^{1/2}.
  \end{equation*}
  By expanding in a Neumann series, we see that $1 + J_\ell
  X_\ell^{C+E} +R_\ell^E$ has a bounded inverse provided that
  \begin{equation}
    \label{eq:B_bound}
    \| (1+J_\ell X_\ell^{C+E})^{-1}\|\,\|R_\ell^E\| < 1.
  \end{equation}
  We first examine $\| (1+J_\ell X_\ell^E)^{-1}\|$.  We have
  \begin{equation*}
    \frac{1}{1+J_\ell X_\ell^E} = 1 - J_\ell
    (X_\ell^E)^{1/2} \frac{1}{1+(X_\ell^E)^{1/2}J_\ell (X_\ell^E)^{1/2}}
    (X_\ell^E)^{1/2}\,,
  \end{equation*}
  and thus
  \begin{equation*}
    \left\|\frac{1}{1+J_\ell X_\ell^E}\right\|
    \leq 1 + \|X_\ell^E\|
    \left\|\frac{1}{1+(X_\ell^E)^{1/2}J_\ell (X_\ell^E)^{1/2}}
    \right\|\,.
  \end{equation*}
  Using the fact that $(4\pi|x-y|)^{-1}\ee^{-\sqrt{E}|x-y|}$ is the
  integral kernel of the operator $\frac{1}{p^2+E}$ for $E\geq 0$, the
  Hardy-Littlewood-Sobolev inequality \cite[Thm. 4.3]{LL} implies that
  $\|X_\ell^E\| \leq \|X_\ell^0\|_2 \leq c_2 \|V_\ell\|_{3/2}$.
  Moreover, $\left\|(1+(X_\ell^E)^{1/2}J_\ell
    (X_\ell^E)^{1/2})^{-1}\right\|$ is the inverse of the eigenvalue
   of $1+(X_\ell^E)^{1/2}J_\ell
  (X_\ell^E)^{1/2}$ with smallest modulus, and this latter operator is isospectral to $1+J_\ell X_\ell^E$.  We
  conclude that $\left\|(1+(X_\ell^E)^{1/2}J_\ell
    (X_\ell^E)^{1/2})^{-1}\right\| \leq {e_\ell(E)}^{-1}$, where
  $e_\ell(E)$ is the smallest eigenvalue of $1+J_\ell X_\ell^E$. The
  latter is bigger than $e_\ell(0)$, which is of order $O(\ell)$ by assumption
  \ref{ax:A}. This shows that there is a constant $c_1 > 0$ such that
  \begin{equation*}
    \left\|\frac{1}{1+J_\ell X_\ell^{C+E}}\right\|
    \leq c_1\ell^{-1}
  \end{equation*}
for small $\ell$. 
  
It remains to bound the operator $ R_\ell^E$, whose integral kernel is given by 
  \begin{align*}
    R_\ell^E(x,y) = \frac{1}{(2\pi)^3}V_\ell(x)^{1/2}|V_\ell(y)|^{1/2}
    \int_{\mathbb{R}^3} \left(\frac{1}{p^2-|p|^b + C +E} -
      \frac{1}{p^2+C+E}\right)\ee^{-i p \cdot (y-x)} \ddd{3}p.
  \end{align*}
 The trace norm $\|R_\ell\|_1$ equals
  \begin{align*}
    \|R_\ell^E\|_1 = \frac{4\pi }{(2\pi)^3}\|V_\ell\|_1 \int_{0}^\infty  \frac{p^b}{p^2-p^b + C +E} \dd p\,.
  \end{align*}
  By dominated convergence, the integral tends to $0$ as $C \to
  \infty$. By H\"older's inequality, $\|V_\ell\|_1 \leq O(\ell)$, so there exists a $C$ such that
  $\|R_\ell^E\| < c_1\ell$.  This shows \eqref{eq:B_bound}.
\end{proof}

Next we show that validity of {\ref{ax:A}}. 
Let
$J = \big\{
  \begin{smallmatrix}
    1,& V \geq 0\\
    -1,& V <0
  \end{smallmatrix}$,
$X = |V|^{1/2} \tfrac{1}{p^2}
|V|^{1/2}$ and $P =
\frac{1}{\langle J\phi|\phi\rangle}
|\phi\rangle \langle J\phi|$ the projection onto the eigenfunction $\phi$ corresponding to the  zero
eigenvalue of $1+V^{1/2} \tfrac{1}{p^2}
|V|^{1/2} = 1 +JX$. Using the unitary scaling operator $U_\ell$ we also introduce the scaled versions $\phi_\ell = U_\ell \phi$, 
$J_\ell = U_\ell J U_\ell^{-1}$ and $P_\ell = U_\ell P U_\ell^{-1}$,
and let $X_\ell = \lambda(\ell) U_\ell X U_\ell^{-1}$.
Then $\langle J_\ell\phi_\ell|\phi_\ell\rangle
= \langle J\phi|\phi\rangle
= -\langle X \phi|\phi\rangle
<0$ does not vanish, since $X = |V|^{1/2} \tfrac{1}{p^2}
|V|^{1/2}$ is a positive operator whose kernel does not contain 
 $\phi$.
Note that $[JX,P] = 0$, which follows from
\begin{equation*}
  (1+JX)P = 0
\end{equation*}
and
\begin{equation*}
  P(1+JX) = J P^* J(1+JX) = JP^*(1+JX)^* J = J\big((1+JX)P\big)^* J =
  0.
\end{equation*}

We decompose $1+J_\ell X_\ell$ as
\begin{equation*}
  1+J_\ell X_\ell = (1+J_\ell X_\ell)P_\ell + (1+J_\ell X_\ell)(1-P_\ell)
\end{equation*}
and examine the two parts separately.
For the first summand note that since 
$JX \phi = -\phi$, 
\begin{equation*}
  (1+J_\ell X_\ell) P_\ell
  = \big(1-\lambda(\ell)\big)P_\ell \,.
\end{equation*}

Next, we study the operator $J_\ell X_\ell (1-P_\ell)$. 
The operator
\begin{equation*}
  T = 1+ J X (1-P) = 1 + J X + P
\end{equation*}
has no zero eigenvalue.
Indeed, if $T \psi = 0$, then
\begin{equation*}
  0 = (1 + J X + P)\big(P + (1-P)\big)\psi = P\psi + (1-P)(1 + J X)\psi,
\end{equation*}
where we used that $P$ commutes with $1+JX$.
Projecting onto $P$ and $1-P$, respectively, yields
\begin{align*}
  0 &= P\psi, & 0 &= (1-P)(1 + J X)\psi = (1+J X)\psi,
\end{align*}
which constrains $\psi$ to be $0$.

Due to the compactness of $P + J X$,
eigenvalues of $T$ can only accumulate at $1$, and hence $T$ has a
bounded inverse $T^{-1}$.  Now
$J_\ell X_\ell =
\lambda(\ell) U_\ell JX U_\ell^{-1}$, and we have the decomposition
\begin{equation*}
  U_\ell^{-1}(1 + J_\ell X_\ell) U_\ell
  = 1 +\lambda(\ell)JX
  = \big(1-\lambda(\ell)\big)P + \big[1+\lambda(\ell)(T-1)\big](1-P)
\end{equation*}
with inverse
\begin{equation*}
  U_\ell^{-1}\frac{1}{1 + J_\ell X_\ell} U_\ell
  = \frac{1}{1 -\lambda(\ell)}P +
  \frac{1}{1-\lambda(\ell)+\lambda(\ell)T}(1-P).
\end{equation*}
This shows that
\begin{equation*}
  \frac{1}{1 + J_\ell X_\ell}(1-P_\ell) = 
  U_\ell\frac{1}{1-\lambda(\ell)+\lambda(\ell)T}(1-P)U_\ell^{-1}.
\end{equation*}
Since $0$, and thus also a
neighborhood of $0$, is not in the spectrum of $T$, and $\lambda(\ell) \rightarrow 1$ as $\ell\to 0$, 
we conclude that $\frac{1}{1-\lambda(\ell)+\lambda(\ell)T}$ is uniformly bounded
for small $\ell$. This yields the uniform boundedness of $\frac{1}{1 + J_\ell
  X_\ell}(1-P_\ell)$.

In order to prove {\ref{ax:a}}
we decompose $\frac{1}{1+J_\ell X_\ell}$ as
\begin{equation*}
  1+J_\ell X_\ell = \frac{1}{1+J_\ell X_\ell} P_\ell + \frac{1}{1+J_\ell
  X_\ell}(1-P_\ell)
  = \frac{1}{1-\lambda(\ell)} P_\ell + \frac{1}{1+J_\ell
  X_\ell}(1-P_\ell),.
\end{equation*}
We have just shown that the second summand is uniformly
bounded in $\ell$.
This allows us to calculate the limit $\ell\to 0$ of the scattering
length $a(V_\ell)$, which equals 
\begin{align*}
  4\pi a(V_\ell) &= \left\langle |V_\ell|^{1/2}\left| \frac{1}{1+V_\ell^{1/2}
    \frac{1}{p^2}|V_\ell|^{1/2}} \right. V_\ell^{1/2}\right\rangle\\
  &= \frac{1}{1 -\lambda(\ell)} \langle |V_\ell|^{1/2}| P_\ell
  V_\ell^{1/2}\rangle + \left\langle |V_\ell|^{1/2}\left| \frac{1}{1+J_\ell
  X_\ell}(1-P_\ell)\right. V_\ell^{1/2}\right\rangle\,.
\end{align*}
Using the uniform boundedness of the second summand together with the
fact that $\|V_\ell\|_1 \to 0$ as
$\ell\to 0$, we see that the second summand vanishes in the limit $\ell\to 0$.
Therefore
\begin{align*}
  \lim_{\ell\to 0}4\pi a(V_\ell)
  &= \lim_{\ell\to 0}\frac{1}{1-\lambda(\ell)}
  \frac{|\langle |V_\ell|^{1/2}| \phi_\ell\rangle|^2}{\langle J_\ell \phi_\ell|\phi_\ell\rangle} \\
  &= \lim_{\ell\to 0}\sqrt{\lambda(\ell)}\frac{\ell}{1-\lambda(\ell)}
    \frac{|\langle |V|^{1/2}| \phi\rangle|^2}{\langle J \phi|\phi\rangle}\\
  &= -\frac{1}{\lambda'(0)} \frac{|\langle |V|^{1/2}| \phi\rangle|^2}{\langle J \phi|\phi\rangle}\,.
\end{align*}

We are left with demonstrating {\ref{ax:n}}. 
This is immediate, since 
\begin{equation*}
  \langle |V_\ell|^{1/2}| \phi_\ell\rangle
  = \sqrt{\lambda(\ell)} \ell^{1/2} \langle |V|^{1/2}| |\phi|\rangle
\end{equation*}
and 
\begin{equation*}
  \langle \phi_\ell| \sgn(V_\ell)\phi_\ell\rangle
  = \langle \phi| \sgn(V) \phi\rangle.
\end{equation*}

\subsection{Example 2}
\label{sec:ex:2}

We consider a sequence of potentials as suggested in \cite{Leggett}, of the form 
\begin{equation}
  \label{eq:potential}
  V_\ell = V_\ell^+ - V_\ell^-, \qquad
  \begin{array}{rcll}
    V_\ell^+(x) &=& (k^+_\ell)^2\ \chi_{\{|x| < \epsilon_\ell\}}(x), &
    \qquad k^+_\ell = k^+ \epsilon_\ell^{-3/2}\\
    V_\ell^-(x) &=&
    (k^-_\ell)^2\ \chi_{\{\epsilon_\ell < |x| < \ell\}}(x), & \qquad k^-_\ell
    = \frac{\frac{\pi}{2}-\ell\omega}{\ell-\epsilon_\ell}\,,
  \end{array}
\end{equation}
with $\omega > 0$, $k^+>0$ and  $0< \epsilon_\ell <  c\ell^2$ with $c<2\omega/\pi$.  The function
$\chi_A(x)$ denotes the characteristic function of the set $A$. (See the sketch on page~\pageref{fig2}.)
We shall show that this sequence of potentials satisfies assumptions \ref{ax:1}--\ref{ax:n}.

Assumptions \ref{ax:1}, \ref{ax:supp},
  \ref{ax:L1}, \ref{ax:positivity} and
  \ref{ax:L2} are in fact obvious, and $\mathcal{V} = \lim_{\ell\to 0} \hat V_\ell(0) = \sqrt{2/\pi} (k^+)^2/3$.

\paragraph{\ref{ax:a}}
To calculate the scattering length $a(V_\ell)$, we have to find the
solution $\psi_\ell$ of $-\Delta \psi_\ell + V_\ell \psi_\ell = 0$,
with $\lim_{|x| \to \infty} \psi_\ell(x) = 1$. The scattering length
then appears in the asymptotics
\begin{equation*}
  \psi_\ell(x) \approx 1 - \frac{a(V_\ell)}{|x|}
\end{equation*}
for large $|x|$. 
To solve the zero-energy scattering equation, we write $\psi_\ell(x) =
\frac{u_\ell(|x|)}{|x|}$ with $u_\ell(0)=0$. Then $u_\ell$ solves the equation
\begin{equation*}
  -u_\ell'' + V_\ell u_\ell = 0.
\end{equation*}
For $r\geq \ell$ the function $u_\ell$ is of the form $u_\ell(r) = c_1
r + c_2$.  The normalization at infinity requires $c_1 = 1$, and
$\psi_\ell$ automatically has the desired asymptotics with $a(V_\ell)
= -c_2$.

In our example, the equation we have to solve is 
\begin{equation*}
  \begin{cases}
    -u_\ell''(r)+(k^+_\ell)^2u_\ell(r) = 0, & 0\leq r \leq \epsilon_\ell\\
    -u_\ell''(r)-(k^-_\ell)^2u_\ell(r) = 0, & \epsilon_\ell \leq r
    \leq \ell\\
    -u_\ell''(r) = 0, & r \geq \ell,
  \end{cases}
\end{equation*}
with the solution
\begin{equation*}
  u_\ell(r) =
  \begin{cases}
    A \sinh(k^+_\ell r), & 0\leq r \leq \epsilon_\ell\\
    B_1 \cos(k^-_\ell r) + B_2 \sin(k^-_\ell r), & \epsilon_\ell \leq
    r
    \leq \ell \\
    r - a(V_\ell),& r \geq \ell.
  \end{cases}
\end{equation*}
Continuity of $u_\ell$ and $u_\ell'$ then requires
\begin{align*}
  A\sinh(k^+_\ell \epsilon_\ell) &= B_1 \cos(k^-_\ell
  \epsilon_\ell) + B_2 \sin(k^-_\ell \epsilon_\ell) \\
  A k^+_\ell \cosh(k^+_\ell \epsilon_\ell) &= -B_1 k^-_\ell
  \sin(k^-_\ell \epsilon_\ell) + B_2 k^-_\ell \cos(k^-_\ell
  \epsilon_\ell)
\end{align*}
and
\begin{align*}
  B_1 \cos(k^-_\ell
  \ell) + B_2 \sin(k^-_\ell \ell) &= \ell - a(V_\ell)\\
  -B_1 k^-_\ell \sin(k^-_\ell \ell) + B_2 k^-_\ell
  \cos(k^-_\ell \ell) &= 1.
\end{align*}
Solving for $a(V_\ell)$ yields
\begin{equation}
  \label{eq:ex2:a}
  a(V_\ell) =
  \ell - \frac{1}{k^-_\ell} \frac{ k^+_\ell
    \tan\big(k_\ell^-(\ell-\epsilon_\ell)\big)+ k^-_\ell\tanh(k^+_\ell\epsilon_\ell) }
  {k^+_\ell- k^-_\ell\tan\big(k_\ell^- (\ell-\epsilon_\ell)\big)
    \tanh(k^+_\ell\epsilon_\ell) }\,.
\end{equation}

By Eq.~\eqref{eq:potential}, $(\ell-\epsilon_\ell) k_\ell^-=
\frac{\pi}{2} - \ell\omega$ and $k^+_\ell \epsilon_\ell = k^+
\epsilon_\ell^{-1/2}$. Since we assume that $\epsilon_\ell = O(\ell^2)$, we thus obtain as expression for the scattering
length in the limit $\ell \to 0$
\begin{equation*}
  \lim_{\ell\to 0} a(V_\ell) = -\lim_{\ell\to 0}
  \frac{\tan\big(\frac{\pi}{2} - \ell\omega\big)}{k^-_\ell}
  = -\frac{2}{\pi\omega}\,.
\end{equation*}
This shows the validity of \ref{ax:a}.

\paragraph{\ref{ax:3}}
To verify assumption \ref{ax:3}, we have to compute the Fourier transform of $V_\ell$, which equals 
\begin{equation*}
  \hat{V}_\ell(p) = \sqrt{\frac{2}{\pi}}\Big[ \epsilon_\ell^3 \big((k_\ell^+)^2+(k_\ell^-)^2\big)\ \varsigma(|p|\epsilon_\ell)
    -(k_\ell^-)^2\ell^3\ \varsigma(|p|\ell)\Big],
\end{equation*}
with  $\varsigma(x) = \frac{1}{x^3}\big(\sin(x)- x\cos(x)\big)$.
Since $|\varsigma(p)| \leq \varsigma(0)=1/3$, one  readily checks that $|\hat V_\ell(p)|\leq 2 |\hat V_\ell(0)|$ for $\ell$ small enough.

\paragraph{\ref{ax:infspec}}
Our next goal is to verify assumption \ref{ax:infspec}.  Let $U(x) =
\frac{\pi^2}{4} \chi_{|x|\leq 1}(x)$ and set $U_\ell(x) =
\ell^{-2}U(x/\ell)$.  For $\lambda(\ell) =
\Bigl(\frac{1-2\frac{\omega}{\pi}\ell}{1-\epsilon_\ell/\ell}\Bigr)^2$,
the potential $W_\ell(x) = \lambda(\ell) U_\ell(x)$ agrees with
$ V^-_\ell(x)$ on its support, so obviously $-W_\ell(x) \leq -
V^-_\ell(x)\leq V_\ell(x)$ holds. The function $U_\ell(x)$ is chosen
such that $p^2 - U_\ell(x)$ has a zero energy resonance.  Indeed,
\begin{equation*}
  \psi(x) =
  \begin{cases}
    \sin(\frac{\pi}{2}|x|)/|x|,& |x| \leq 1 \\
    1/|x|,& |x| \geq 1
  \end{cases}
  \in L^2_{\textrm{loc}}(\mathbb{R}^3) \setminus L^2(\mathbb{R}^3)
\end{equation*}
is a generalized eigenfunction of $p^2-U$ and $\psi_\ell(x) =
\psi(x/\ell)$ is a generalized eigenfunction of $p^2-U_\ell$.
Therefore, $U_\ell^{1/2}\frac{1}{p^2}U_\ell^{1/2}$ has the
eigenvector $U_\ell^{1/2} \psi_\ell \in
L^2(\mathbb{R}^3)$ to the eigenvalue $1$.
  
Note that our condition on $\epsilon_\ell$ implies that $\lambda(\ell) < 1 - c \ell$ for some constant $c>0$ and 
small enough $\ell$. Since $V_\ell^- \leq W_\ell$, the largest eigenvalue of
$(V^-_\ell)^{1/2}\frac{1}{p^2} (V^-_\ell)^{1/2}$ is smaller or equal to
$\lambda(\ell)$, i.e., 
\begin{equation}
  \label{eq:resonance}
  \left\| (V^-_\ell)^{1/2}\frac{1}{p^2} (V^-_\ell)^{1/2}
  \right\|
  \leq \lambda(\ell)
  \leq 1 - c\ell\,.
\end{equation}
Now choose $C > 0$ such that $p^2 -|p|^b + C > 0$ and define the operator
\begin{align*}
  R_\ell=(V^-_\ell)^{1/2} \frac{1}{p^2-|p|^b + C} (V^-_\ell)^{1/2} -
  (V^-_\ell)^{1/2} \frac{1}{p^2+C} (V^-_\ell)^{1/2} \,.
\end{align*}
Its trace norm equals
\begin{align*}
  \|R_\ell\|_1 = \frac{1}{2\pi^2}\|V^-_\ell\|_1  \int_0^\infty
  \frac{p^2}{p^2+C} \frac{p^b}{p^2-p^b+C} \dd p\,,
\end{align*}
which tends to zero as $C\to \infty$ by monotone convergence. Since $\|V^-_\ell\|_1 = O(\ell)$,  there is a 
$C$ such that $\|R_\ell\| < c\ell$, proving that
\begin{align*}
  \left\| (V^-_\ell)^{1/2} \frac{1}{p^2-|p|^b + C} (V^-_\ell)^{1/2}
  \right\| &\leq \left\| (V^-_\ell)^{1/2} \frac{1}{p^2+C}
    (V^-_\ell)^{1/2}\right\| +
  \|R_\ell\|\\
  &< \left\| (V^-_\ell)^{1/2} \frac{1}{p^2} (V^-_\ell)^{1/2}\right\| +
  c\ell \leq 1,
\end{align*}
where we have used \eqref{eq:resonance} in the last step. By the
Birman-Schwinger principle, this shows that $p^2-V^-_\ell - |p|^b +
C\geq 0$, and hence also $p^2+V_\ell - |p|^b +C\geq 0$.

\paragraph{\ref{ax:A}}
Note that $V_\ell^{1/2}
\frac{1}{p^2}|V_\ell|^{1/2}$ has an eigenvalue $-\lambda^{-1}\neq 0$ if and only if $p^2 + \lambda V_\ell $ as a zero-energy resonance. Equivalently, the scattering length $a(\lambda V_\ell)$ diverges. According to our calculation \eqref{eq:ex2:a}, this happens for $\lambda>0$ either satisfying 
$$
k^+_\ell =  k^-_\ell\tan\big(\sqrt{\lambda}k_\ell^- (\ell-\epsilon_\ell)\big)
    \tanh(\sqrt{\lambda}k^+_\ell\epsilon_\ell) 
$$
or $\sqrt{\lambda}k_\ell^- (\ell -\epsilon_\ell) = m \pi /2$ for odd integer $m$. 
The smallest $\lambda$ satisfying either of these equations is $\lambda =  1 + 4 \ell\omega/\pi + O(\ell^2)$, 
hence the smallest eigenvalue of  $1+ V_\ell^{1/2}
\frac{1}{p^2}|V_\ell|^{1/2}$ is $$e_\ell =  4 \ell\omega/\pi + O(\ell^2).$$

We are left with showing that  
$$\left(1+
  V_\ell^{1/2}\frac{1}{p^2}|V_\ell|^{1/2}\right)^{-1}(1-P_\ell)
  $$ 
  is uniformly
bounded in $\ell$. This follows directly from \cite[Consequence 1]{BHS-delta}. 
For the sake of completeness we repeat the argument here. 

First, recall that $\phi_\ell$ denotes the eigenvector of $1+ V_\ell^{1/2}
\frac{1}{p^2}|V_\ell|^{1/2}$ to its smallest eigenvalue $e_\ell$, and 
$J_\ell = \big\{
  \begin{smallmatrix}
    1,& V_\ell \geq 0\\
    -1,& V_\ell <0
  \end{smallmatrix}$. We also introduce the notation $X_\ell =
|V_\ell|^{1/2}\frac{1}{p^2}|V_\ell|^{1/2}$ and
$X^\pm_\ell = |V^\pm_\ell|^{1/2}\frac{1}{p^2}|V^\pm_\ell|^{1/2}$.

We now pick some $\psi\in L^2(\mathbb{R}^3)$ and set
  \begin{equation}
    \label{eq:varphi}
    \varphi =(1+
  V_\ell^{1/2}\frac{1}{p^2}|V_\ell|^{1/2})^{-1}(1-P_\ell) \psi= \frac{1}{1+J_\ell
      X_\ell}(1-P_\ell)\psi = \frac{1}{J_\ell+
      X_\ell}J_\ell(1-P_\ell)\psi\,.
  \end{equation}
  Below we are going to show that there exists a constant $c>0$ such that for small
  enough $\ell$
  \begin{equation}
    \label{eq:aah}
    \langle \varphi | (1-X_\ell^-) \varphi \rangle 
    \geq c \|\varphi\|^2_{L^2}.
  \end{equation}
In order to utilize this inequality we need the following lemma, which already appeared in \cite[Lemma~1]{BHS-delta}.
\begin{lemma}
  \label{lemma:BS-bound}
  Let $V = V_+ - V_-$, where $V_-, V_+\geq 0$ have disjoint
  support. Denote $J = \big\{
  \begin{smallmatrix}
    1,& V \geq 0\\
    -1,& V <0
  \end{smallmatrix}$,$\quad$ $X =
  |V|^{1/2}\frac{1}{p^2}|V|^{1/2}$ and $X_\pm =
  V_\pm^{1/2}\frac{1}{p^2}V_\pm^{1/2}$. Then for any $\phi \in
  L^2(\mathbb{R}^3)$, we have
  \begin{equation}
    \label{eq:B-S-inequality}
    \sqrt{2}\  \|\phi\| \| (J+X)\phi \| \geq \langle \phi| (X_+ +1 - X_-)\phi\rangle.
  \end{equation}
\end{lemma}
\begin{proof}
  Decompose $\phi = \phi_+ + \phi_-$, such that $\supp(\phi_-) \subseteq
  \supp(V_-)$ and $\supp(\phi_+) \cap \supp(V_-) = \emptyset$.
    By applying the Cauchy-Schwarz inequality, we have 
    \begin{align*}
      \|(J+X)\phi\| \|\phi_+\| &\geq \Re \langle \phi_+| (J+X) \phi
      \rangle = \langle \phi_+| (1+X_+) \phi_+\rangle + \Re \langle
      \phi_+| V_+^{1/2}\tfrac{1}{p^2} V_-^{1/2} \phi_-\rangle,
      \\
      \|(J+X)\phi\| \|\phi_-\| &\geq \Re \langle (J+X) \phi|
      -\phi_-\rangle = \langle \phi_-| (1-X_-) \phi_-\rangle -\Re \langle
      \phi_+| V_+^{1/2}\tfrac{1}{p^2} V_-^{1/2} \phi_-\rangle.
    \end{align*}
    We add the two inequalities and obtain
    \begin{equation*}
      \|(J+X)\phi \|\ \big(\|\phi_+\| + \|\phi_-\|\big)
      \geq
      \langle \phi_+| (1+X_+) \phi_+\rangle
      + \langle \phi_-| (1-X_-) \phi_-\rangle
      = \langle \phi| (X_+ +1 - X_-)\phi\rangle.
    \end{equation*}
    Finally, we use that $\|\phi_+\| + \|\phi_-\| \leq \sqrt{2} \|\phi\|$,
    which completes the proof.
  \end{proof}
In combination with Lemma~\ref{lemma:BS-bound} the inequality \eqref{eq:aah} immediately yields 
  \begin{align*}
    \sqrt{2} \|\varphi\|\|(J_\ell+X_\ell)\varphi\| &\geq \langle
    \varphi| (1-X_\ell^-) \varphi \rangle \geq c \|\varphi\|^2,
  \end{align*}
  which further implies that
  \begin{equation*}
    \|\psi\| \geq \|J_\ell(1-P_\ell)\psi\|
    = \|(J_\ell+X_\ell)\varphi\|
    \geq
    \frac c {\sqrt{2}} \|\varphi\| = \frac c {\sqrt{2}} \|(1+J_\ell X_\ell)^{-1}(1-P_\ell)\psi\|\,,
  \end{equation*}
  proving uniform boundedness of $(1+
  V_\ell^{1/2}\frac{1}{p^2}|V_\ell|^{1/2})^{-1}(1-P_\ell)$.
  
 It remains to show the inequality \eqref{eq:aah}. 
  To this aim we denote by $\phi_\ell^-$ the eigenvector corresponding to the smallest eigenvalue  $e_\ell^->0$ of $1-X_\ell^-$ and by $P_{\phi_\ell^-}$ the orthogonal
  projection onto $\phi_\ell^-$.  The Birman-Schwinger operator
  $X_\ell^-$ corresponding to the potential $V_\ell^-$ has only one
  eigenvalue close to $1$. All other
  eigenvalues are separated from $1$ by a gap of order one. Hence
  there exists $c_1 > 0$ such that
  \begin{equation*}
    (1-X_\ell^-)(1-P_{\phi_\ell^-}) \geq c_1
  \end{equation*}
  and, therefore,
  \begin{align*}
    \langle \varphi | (1-X_\ell^-) \varphi \rangle
    &\geq c_1 \langle \varphi | (1-P_{\phi^-_\ell}) \varphi \rangle 
    + e_\ell^- \langle \varphi | P_{\phi^-_\ell} \varphi \rangle\\
    &= c_1 \| \varphi \|_{L^2}^2 + (e_\ell^- - c_1)
    \langle \varphi | P_{\phi^-_\ell} \varphi \rangle.
  \end{align*}
 With $P_{J_\ell \phi_\ell} = | J_\ell \phi_\ell\rangle \langle
  J_\ell \phi_\ell|$ being the orthogonal projection onto
  $J_\ell\phi_\ell$ we can write 
  \begin{equation*}
    \varphi = (1-P_{J_\ell\phi_\ell})\varphi,
  \end{equation*}
 simply for the reason that, because of \eqref{eq:varphi} and the fact that $P_\ell$ commutes with $B_\ell$,
  \begin{align*}
P_{J_\ell \phi_\ell} \varphi =     P_{J_\ell \phi_\ell} (1+J_\ell X_\ell)^{-1}(1-P_\ell) \psi= P_{J_\ell
      \phi_\ell} (1-P_\ell) (1+J_\ell X_\ell)^{-1} \psi= 0\,.
  \end{align*}
  Consequently,
  \begin{align*}
    |\langle \varphi | P_{\phi^-_\ell} \varphi \rangle|
    &=
    |\langle \varphi | (1-P_{J_\ell\phi_\ell}) P_{\phi^-_\ell} \varphi \rangle|
    \leq
    \|\varphi\|^2_{L^2} \|(1-P_{J_\ell\phi_\ell})P_{\phi^-_\ell} \|\\
    &=
    \|\varphi\|^2_{L^2}\|(1-P_{J_\ell\phi_\ell})\phi_\ell^-\|^2
    = 
    \|\varphi\|^2_{L^2} \|(1-P_{\phi_\ell^-})J_\ell\phi_\ell\|^2\,.
  \end{align*}
  To estimate $\|(1-P_{\phi_\ell^-})J_\ell\phi_\ell\|$, we apply
  Lemma~\ref{lemma:BS-bound} to $\phi_\ell$ and obtain
  \begin{align*}
    \sqrt{2}\, e_\ell = \sqrt{2}\, \|(J_\ell + X_\ell)\phi_\ell\|
    &\geq \langle \phi_\ell|(1-X_\ell^-)\phi_\ell\rangle
    = \langle J_\ell \phi_\ell|(1-X_\ell^-)J_\ell \phi_\ell\rangle \\
    &= e^-_\ell|\langle J_\ell\phi_\ell|\phi_\ell^-\rangle|^2 +
    \langle (1-P_{\phi_\ell^-})J_\ell
    \phi_\ell|(1-X_\ell^-)(1-P_{\phi_\ell^-})J_\ell \phi_\ell\rangle\\
    &\geq
    c_1\|(1-P_{\phi_\ell^-})J_\ell \phi_\ell\|^2\,.
  \end{align*}
This shows that 
  $\|(1-P_{\phi_\ell^-})J_\ell \phi_\ell\| = O(\ell^{1/2})$ and consequently \eqref{eq:aah} holds for small enough $\ell$.

\paragraph{\ref{ax:n}}
By construction, $1- (V_\ell^-)^{1/2}
\frac{1}{p^2}(V_\ell^-)^{1/2}$ has no negative eigenvalues. By applying Lemma~\ref{lemma:BS-bound} to $\phi_\ell$, we obtain
\begin{equation}
  \label{eq:lim:X}
  \left\langle \phi_\ell \left| (V_\ell^+)^{1/2}
  \tfrac{1}{p^2}(V_\ell^+)^{1/2}\right. \phi_\ell\right\rangle \leq \sqrt{2} e_\ell\quad \text{and} \quad
  \left\langle \phi_\ell\left| \left(1- (V_\ell^-)^{1/2}
  \frac{1}{p^2}(V_\ell^-)^{1/2}\right)\right. \phi_\ell\right\rangle \leq \sqrt{2} e_\ell\,.
\end{equation}
We claim that this implies
\begin{equation*}
  \lim_{\ell \to 0}\langle J_\ell \phi_\ell|\phi_\ell\rangle = -1.
\end{equation*}
Indeed,
\begin{equation*}
  (J_\ell + X_\ell) \phi_\ell = e_\ell J_\ell \phi_\ell
\end{equation*}
and thus
\begin{align*}
  (1-e_\ell) \langle J_\ell \phi_\ell|\phi_\ell\rangle
  &= -\langle \phi_\ell| X_\ell \phi_\ell \rangle\\
  &= -\langle \phi_\ell| X^+_\ell \phi_\ell\rangle
  -\langle \phi_\ell| X^-_\ell \phi_\ell \rangle\\
  &\phantom{=\ } -\langle \phi_\ell| (V_\ell^-)^{1/2}
  \tfrac{1}{p^2}(V_\ell^+)^{1/2}\phi_\ell \rangle -\langle \phi_\ell|
  (V_\ell^+)^{1/2} \tfrac{1}{p^2}(V_\ell^-)^{1/2}\phi_\ell \rangle.
\end{align*}
Adding $1$ on both sides yields
\begin{align*}
  (1-e_\ell) \bigl< (1+J_\ell) \phi_\ell|\phi_\ell\bigr> + e_\ell &=
  -\langle \phi_\ell| X^+_\ell \phi_\ell\rangle
  +\langle \phi_\ell| (1 - X^-_\ell)\phi_\ell \rangle\\
  &\phantom{=\ } -\langle \phi_\ell| (V_\ell^-)^{1/2}
  \tfrac{1}{p^2}(V_\ell^+)^{1/2}\phi_\ell \rangle -\langle \phi_\ell|
  (V_\ell^+)^{1/2} \tfrac{1}{p^2}(V_\ell^-)^{1/2}\phi_\ell \rangle.
\end{align*}
By taking the absolute value, applying the Cauchy-Schwarz inequality and using
\eqref{eq:lim:X}, we obtain
\begin{equation}
  \label{eq:J:expectation}
    \big| \bigl< (1+J_\ell) \phi_\ell|\phi_\ell\bigr>\big| \leq
    \frac{1}{1-e_\ell} \left( (1+2\sqrt{2}) e_\ell + 2 \sqrt{\langle \phi_\ell|
        X^+_\ell \phi_\ell\rangle \langle \phi_\ell| X^-_\ell
        \phi_\ell\rangle}
    \right) = O( e_\ell^{1/2})\,.
\end{equation}

Finally, to bound $\langle |V_\ell|^{1/2}||\phi_\ell|\rangle$, we note that $\langle |V^-_\ell|^{1/2}||\phi_\ell|\rangle\leq \|V_\ell^-\|_1^{1/2} = O(\ell^{1/2})$. For the analogous bound with $V_\ell^-$ replaced by $V^+_\ell$, we can again employ Lemma~\ref{lemma:BS-bound}, which implies that $\sqrt{2} e_\ell \geq \langle \phi^+_\ell|X_\ell^++1|\phi_\ell^+\rangle \geq \|\phi_\ell^+\|_2^2$ (where $\phi^+_\ell = \frac 12(1+J_\ell)\phi_\ell$), hence  $\langle |V^+_\ell|^{1/2}||\phi_\ell|\rangle\leq \|V_\ell^+\|_1^{1/2} \|\phi_\ell^+\|_2 \leq O(\ell^{1/2})$. This completes the proof.

\section{The Definition of $T_c$}\label{sec:appendix:tc}

In this appendix we shall show that the equations~\eqref{eq:T_c}  define $T_c$ and $\tilde\mu$ uniquely. To start, let $F:\R\times \R_+\to \R^2$ be defined by its components 
  \begin{equation}
 F_1(\nu,T) = \frac 1{(2\pi)^{3}} \int_{\mathbb{R}^3}  \left(
        \frac{\tanh\big(\frac{p^2-\nu}{2T}\big)}{p^2-\nu}
        -\frac{1}{p^2} \right) \ddd{3}p
 \end{equation}
and
\begin{equation}
F_2(\nu,T) = \nu + \frac{2
        \mathcal{V}}{(2\pi)^{3/2}}\int_{\mathbb{R}^3}
      \frac{1}{1+\ee^{\frac{p^2-\nu}{T}}} \ddd{3}p\,.
  \end{equation}
We clearly have $\partial F_1/\partial T < 0$ and $\partial F_2/\partial\nu > 0$ (since $\mathcal{V}\geq 0$ by assumption). 
 By
  dominated convergence, we may interchange the derivative with the
  integral and compute
\begin{equation}
\frac{\partial F_1}{\partial \nu} =  \frac 1{(2\pi)^{3}}  \int_{\mathbb{R}^3}
    \frac{\kappa'(\frac{p^2-\nu}{2T})}{\kappa(\frac{p^2-\nu}{2T})^2}
    \ddd{3}p \,,
\end{equation}
where $\kappa(x)= x/\tanh(x)$. If $\nu\leq 0$, this is positive, since $\kappa'(t)\geq 0$ for $t\geq 0$. If $\nu > 0$, on the other hand, we can integrate out the angular coordinates and change variables to $\pm
  t = p^2-\nu$, respectively, to obtain  
  \begin{equation}
   \frac{\partial F_1}{\partial \nu} = \frac 1{4\pi^2} \int_0^\infty \frac{
      \kappa'(\frac{t}{2T})\sqrt{\nu+t}} {
      \kappa^2(\frac{t}{2T})} \dd t -  \frac 1{4\pi^2} \int_0^{\nu} \frac{
      \kappa'(\frac{t}{2T})\sqrt{\nu-t}}
    {\kappa^2(\frac{t}{2T})}
    \dd t \,.
  \end{equation}
  Since $\sqrt{\nu+t} > \sqrt{\nu-t}$, it is clear that this sum is positive, i.e., $\partial F_1/\partial\nu >0$.

 We proceed similarly to show that $\partial F_2/\partial T > 0$. We have 
\begin{align}\nonumber
\frac{\partial F_2}{\partial T} & = \frac{\mathcal{V}}{2 T^2} \frac {1}{(2\pi)^{3/2}} \int_{\R^3} \frac{p^2-\nu}{\cosh\big(\frac{p^2-\nu}{2T}\big)^2} \ddd{3}p   \\ & = \frac{\mathcal{V}}{2 T^2} \frac {1}{\sqrt{2\pi}}  \left( \int_0^\infty \frac{ t \sqrt{t + \nu}}{\cosh\big(\frac t{2T}\big)^2} \dd t -  \int_0^\nu \frac{ t \sqrt{\nu-t}}{\cosh\big(\frac t{2T}\big)^2} \dd t \right) > 0\,.
\end{align}
In particular, the Jacobian determinant of $F$ is strictly positive.

For fixed $T$, we have $\lim_{\nu\to-\infty} F_2(\nu,T) = -\infty$ and $\lim_{\nu\to \infty} F_2(\nu,T)=\infty$. Hence there is a unique solution $\nu_T$ of the equation $F_2(\nu,T)= \mu$, for any $\mu \in \R$, and $\nu_T$ is decreasing in $T$.  Moreover, the function $T\mapsto F_1(\nu_T,T)$ is strictly decreasing, and hence the equation $F_1(\nu_T,T)=\lambda$ has a unique solution for $\lambda$ in its range. In particular, $T_c$ is a strictly decreasing function of $\lambda = -1/(4\pi a)$, hence a strictly decreasing function of $a$ for $a<0$.

For $\mu\leq 0$, one checks that $\lim_{T\to 0} F_1(\nu_T,T) \leq 0$, hence $T_c=0$. For $\mu>0$, however, $\lim_{T\to 0} F_1(\nu_T,T) = \infty$, hence $T_c>0$ for any $a<0$.

\bibliographystyle{abbrv}

\begin{thebibliography}{10}

\bibitem{albeverio}
S.~Albeverio, F.~Gesztesy, R.~H{\o}egh-Krohn, and H.~Holden.
\newblock {\em Solvable models in quantum mechanics.}
\newblock
Second edition. With an appendix by Pavel Exner. AMS Chelsea Publishing, Providence, RI, 2005.

\bibitem{albeverio82}
S.~Albeverio, F.~Gesztesy, R.~H{\o}egh-Krohn. 
\newblock
{\em The low energy expansion in nonrelativistic scattering theory.}  
\newblock
{ Ann. Inst. H. Poincar\'e}
 {\bf 37}, 1--28, 1982.

\bibitem{BacFroJon-09}
V.~Bach, J.~Fr{\"o}hlich, L.~Jonsson,
{\it Bogolubov-{H}artree-{F}ock mean field theory for neutron stars and
  other systems with attractive interactions}, 
J. Math. Phys. {\bf 50}, 102102 (2009).

\bibitem{BLS}
V.~Bach, E.H. Lieb, J.P. Solovej,
{\it Generalized Hartree-Fock theory and the Hubbard model}, 
J. Stat. Phys. {\bf 76}, 3--89 (1994).

\bibitem{bcs}
J.~Bardeen, L.N. Cooper, J.R. Schrieffer,
{\it Theory of superconductivity}, 
Phys. Rev. {\bf 108}, 1175--1204 (1957).

\bibitem{BHS-delta}
G.~Br\"aunlich, C.~Hainzl, R.~Seiringer,
{\it On contact interactions as limits of short-range potentials},
Preprint:  arXiv:1305.3805 

\bibitem{FHNS2007}
R.L. Frank, C.~Hainzl, S.~Naboko, R.~Seiringer, 
{\it The critical temperature for the {BCS} equation at weak coupling}, 
J. Geom. Anal. {\bf 17}, 559--567 (2007).

\bibitem{FHSS-micro_ginzburg_landau}
R.L. Frank, C.~Hainzl, R.~Seiringer, J.P. Solovej, 
{\it Microscopic derivation of Ginzburg--Landau theory}, 
J. Amer. Math. Soc. {\bf 25}, 667--713 (2012).

\bibitem{HHSS}
C.~Hainzl, E.~Hamza, R.~Seiringer, J.P. Solovej,
{\it The {BCS} functional for general pair interactions}, 
Comm. Math. Phys. {\bf 281}, 349--367 (2008).

\bibitem{HaiLenLewSch-10}
C.~Hainzl, E.~Lenzmann, M.~Lewin, B.~Schlein, 
{\it On blowup for time-dependent generalized {H}artree-{F}ock equations}, 
Ann. Henri Poincar{\'e} {\bf 11}, 1023--1052 (2010).

\bibitem{HLS2008}
C.~Hainzl, M.~Lewin, R.~Seiringer, 
{\it A nonlinear model for relativistic electrons at positive temperature}, 
Rev. Math. Phys. {\bf 20}, 1283--1307 (2008).

\bibitem{HS-mu}
C.~Hainzl, R.~Seiringer, 
{\it The {BCS} critical temperature for potentials with negative
  scattering length}, 
Lett. Math. Phys. {\bf 84}, 99--107 (2008).

\bibitem{HS-T_C}
C.~Hainzl, R.~Seiringer, 
{\it Critical temperature and energy gap for the BCS equation}, 
Phys. Rev. B {\bf 77}, 184517 (2008).

\bibitem{HS-spec_prop}
C.~Hainzl, R.~Seiringer, 
{\it Spectral properties of the {BCS} gap equation of superfluidity}, 
in: Mathematical results in quantum mechanics, pp. 117--136, 
World Sci. Publ., Hackensack, NJ (2008). 

\bibitem{Leggett}
A.J. Leggett, 
{\it  Diatomic molecules and Cooper pairs}, 
in: A.~Pekalski and J.~Przystawa (eds.),  Modern Trends in the
  Theory of Condensed Matter, Lecture Notes in Physics {\bf 115}, pp. 13--27,  Springer Berlin / Heidelberg (1980).

\bibitem{leggett_quantum_liquids}
A.J. Leggett, 
{\it Quantum liquids}, 
Science {\bf 319}, 1203--1205 (2008).

\bibitem{LenLew-12b}
M.~{Lewin}, E.~{Lenzmann}, 
{\it Minimizers for the Hartree-Fock-Bogoliubov theory of neutron stars and white dwarfs}, Duke Math. J. {\bf 152}, 257--315 (2010).

\bibitem{LewPau-12}
M.~{Lewin}, S.~{Paul}, 
{\it A Numerical Perspective on Hartree-Fock-Bogoliubov Theory}, 
preprint  arXiv:1206.6081 

\bibitem{LL}
E.H. Lieb, M.~Loss, 
{\it Analysis}, 
American Mathematical Society, $2^{\rm nd}$ edition (2001).

\bibitem{NRS}
P.~Nozi\'eres, S.~Schmitt-Rink, 
{\it Bose condensation in an attractive fermion gas: From weak to strong
  coupling superconductivity}, 
J. Low Temp. Phys. {\bf 59}, 195--211 (1985).

\end{thebibliography}

\end{document}